\documentclass{article} 
\usepackage{iclr2025_conference,times}
\usepackage{adjustbox}

\usepackage{algorithm,amssymb}
\usepackage{algpseudocode,wrapfig}
\newtheorem{lemma}{Lemma}[section]
\newtheorem{theorem}{Theorem}[section]
\newtheorem{corollary}{Corollary}[section] 
\newtheorem{proof}{Proof}[section]

\usepackage{amsmath,amsfonts,bm}

\def\eqref#1{equation~\ref{#1}}

\def\1{\bm{1}}

\DeclareMathAlphabet{\mathsfit}{\encodingdefault}{\sfdefault}{m}{sl}
\SetMathAlphabet{\mathsfit}{bold}{\encodingdefault}{\sfdefault}{bx}{n}

\usepackage{hyperref}
\usepackage{url} 

\newcommand{\goto}{\rightarrow} 
 
\newcommand{\logit}{\mathrm{logit}} 
\newcommand{\diag}{\mathrm{diag}}

\title{Fitting networks with a cancellation trick}
\usepackage{xr}

\author{Jiashun Jin  \\
Department of Statistics\\
Carnegie Mellon University\\
Pittsburgh, PA 15213, USA \\
\texttt{jiashun@andrew.cmu.edu} \\
\And
Jingming Wang \\
Department of Statistics \\
University of Virginia \\
Virginia, VA 22903, USA \\
\texttt{pdw9qv@virginia.edu} 
}

\iclrfinalcopy 
\begin{document}

\maketitle


\begin{abstract} 
The degree-corrected block model (DCBM), latent space model (LSM), and $\beta$-model  
are all popular network models.  We combine their modeling ideas and propose the logit-DCBM as a new model.  Similar as the $\beta$-model and LSM, the logit-DCBM contains nonlinear factors, where fitting the parameters   is a challenging open problem. We resolve this  problem by introducing a cancellation trick. 
We also propose R-SCORE as a recursive community detection algorithm, where in each iteration, we first use the idea above to update 
our parameter estimation,  and then use the results to remove the nonlinear factors in the logit-DCBM so the renormalized model  approximately satisfies a low-rank model, just like the DCBM. 
 Our numerical study suggests that R-SCORE significantly improves over existing spectral approaches in many cases.  Also, theoretically, we show that  the Hamming error rate of R-SCORE is faster 
 than that of SCORE in a specific sparse region, and is at least as fast outside this region.  
 \end{abstract}

\tableofcontents
 
\section{Introduction}\label{sec:intro}
Community detection is a problem of major interest in network analysis (e.g., see \citep{goldenberg2010survey}, a survey paper).    Consider an undirected network 
with $n$ nodes and $K$ communities 
${\cal C}_1, {\cal C}_2, \ldots, {\cal C}_K$ (a community is a group of nodes with similar behaviors).     
Let $A \in \mathbb{R}^{n, n}$ be the adjacency matrix,  where $A_{ij} = 1$ if 
and only if there is an edge between node $i$ and $j$, $1 \leq i \neq j \leq n$.   Conventionally, 
we do not count self edges, so $A_{ij} = 0$ if $i = j$.  As in many works on community detection (e.g., \cite{chen2018convexified,  zhao2012consistency, yuan2022likelihood}),  we assume that each node belongs to exactly one of the $K$ communities.    For each $1 \leq i \leq n$, we encode the community label of node $i$ by a $K$-dimensional 
vector $\pi_i$ (which is unknown to us) such that 
\begin{equation} \label{DCBM1}  
\mbox{$\pi_i= e_k$ \; if and only if node $i \in {\cal C}_k$}   \qquad  \mbox{($e_k$: $k$-th standard Euclidean basis vector of $\mathbb{R}^K$)}.  
\end{equation} 
The goal of community detection is to use $(A, K)$ to cluster all $n$ nodes into $K$ communities/groups.

The degree-corrected block model (DCBM) \citep{DCBM} is a  popular 
network model. Suppose we use a free parameter $\theta_i > 0$ to model the {\it degree heterogeneity}  of node $i$,  $1 \leq i \leq n$. For a non-negative matrix $P \in \mathbb{R}^{K, K}$, DCBM 
assumes that the upper triangular entries of $A$ are independent Bernoulli variables satisfying  
\begin{equation} \label{DCBM2} 
\mathbb{P}(A_{ij} = 1) = \theta_i \theta_j \pi_i' P \pi_j \qquad \Longleftrightarrow 
\qquad \log(\mathbb{P}(A_{ij} = 1))   = \log(\theta_i) + \log(\theta_j) + \pi_i' Q \pi_j, 
\end{equation} 
where $Q$ is a $K \times K$ matrix such that $Q  = \log(P)$ entry-wise.  When all $\theta_i$ are equal,  
DCBM reduces to the well-known {\it Stochastic Block Model (SBM)} \citep{SBM}.    
Note  that as $0 \leq \mathbb{P}(A_{ij} = 1) \leq 1$, so implicitly, DCBM has imposed a 
set of  constraints on its parameters: 
\begin{equation} \label{DCBMconstraint} 
\log(\theta_i) + \log(\theta_j) + \pi_i' Q \pi_j \leq  0,  \qquad 1 \leq i, j \leq n. 
\end{equation} 
These constraints make an already complicated setting even more complicated, 
so we desire to remove them if possible. 
Also, if the matrix $Q$ is positive definite, then $Q = U' U$ for a matrix $U \in \mathbb{R}^{K, K}$. In this special case, we can rewrite (\ref{DCBM2}) as 
\begin{equation} \label{DCBM3} 
\log(\mathbb{P}(A_{ij} = 1))   = \log(\theta_i) + \log(\theta_j) + z_i' z_j,   \qquad \mbox{where $z_i = U \pi_i$, $1 \leq i \leq n$}. 
\end{equation}

The latent space model (LSM)  \citep{Hoff} and the $\beta$-model \citep{betamodel} are also popular network models.  Denote the logit function by $\logit(x) = \log(x/[1-x])$,  $0 < x < 1$. 
In a representative form, for so-called latent positions $z_1, \ldots, z_n \in \mathbb{R}^K$,  the LSM assumes  
\begin{equation} \label{LSM} 
\logit(\mathbb{P}(A_{ij} = 1))  = \log(\theta_i) + \log(\theta_j) + z_i' z_j. 
\end{equation} 
Compared (\ref{LSM}) with (\ref{DCBM3}), the only difference is the log-link is replaced by the logit-link, so at least in the special case where $Q$ is positive definite, two models are similar.  
Also, if we drop the $z_i' z_j$ term on the RHS, then (\ref{LSM}) reduces to the $\beta$ model (where we only have one community, i.e., $K = 1$).  

However, despite the similarity,  to many statisticians, (\ref{LSM}) is highly preferred. The main reason is that, 
for binary data, the model recommended by textbooks is the logistic regression model (e.g., \citep[Section 4.4]{HTF} and  \cite{dobson2018introduction}), where the logit-link function 
was argued to be the most natural. Additionally, some popular Python packages, such as scikit-learn, frequently use the logit-link function.    In fact, in the LSM case,  since $\logit(\mathbb{P}(A_{ij} = 1))$ can take any values in 
$(-\infty, \infty)$, we do not have the constraints (see (\ref{DCBMconstraint}))  as the DCBM case. 

To combine the modeling ideas of  all three models,    we propose the {\it logit-DCBM}, where we assume 
\begin{equation} \label{logit-DCBM} 
\logit(\mathbb{P}(A_{ij} = 1)) = \log(\theta_i) + \log(\theta_j) + \pi_i'  Q  \pi_j,  \;\;\;  \mbox{with $Q = \log(P)$ entrywise as above}. 
\end{equation} 
Since we use the logit-link function,   we do not need the constraints (\ref{DCBMconstraint})  as in the DCBM case.    
Also, we can view (\ref{logit-DCBM}) as an extension of (\ref{DCBM2}).  
Moreover, since we do not require $Q$ to be positive definite in (\ref{logit-DCBM}), so (\ref{logit-DCBM}) also extends (\ref{LSM}) to a broader setting. Last, \ref{logit-DCBM}) reduces to the $\beta$-model   if we let $Q = 0$. 

%
 
In summary, we propose the logit-DCBM as a nonlinear variant of DCBM so hopefully it is more broadly acceptable, 
especially for researchers with a strong preferences in nonlinear network models (such as the LSM) and in using logistic regression type model for binary data.

We now rewrite the logit-DCBM in the matrix form. 
Note that under the model, $\mathbb{P}(A_{ij}  = 1) = N_{ij} \cdot \theta_i \theta_j \pi_i' P \pi_j$, 
where $N_{ij} = [1 + \theta_i \theta_j \pi_i' P \pi_j]^{-1}$ is a nonlinear term. Let $N = (N_{ij})$, 
$\Theta = \diag(\theta_1, \ldots, \theta_n) \in \mathbb{R}^{n, n}$ and 
$\Pi = [\pi_1, \ldots, \pi_n]'$. For any matrix $\Omega \in \mathbb{R}^{n, n}$, let $\diag(\Omega) \in \mathbb{R}^{n, n}$ be the diagonal matrix where the $k$-th diagonal entry is $\Omega_{kk}$.  Let $W \in \mathbb{R}^{n, n}$ be the matrix where $W_{ij} = A_{ij} - \mathbb{E}[A_{ij}]$ if $i \neq j$ and $W_{ij} = 0$ otherwise. Let $\circ$ denote the Hadamard (or entry-wise) product 
\citep{HornJohnson}. 
Under the logit-DCBM model (\ref{logit-DCBM}), 
\begin{equation} \label{logit-DCBM2} 
A = \Omega - \diag(\Omega) + W, \qquad \mbox{with  $\Omega = N  \circ \widetilde{\Omega}$ \; and \; $\widetilde{\Omega} = 
\Theta \Pi P \Pi' \Theta$}. 
\end{equation} 
 Note that $\mathrm{rank}(\widetilde{\Omega}) = K$, but due to the matrix of nonlinear factors $N$, 
$\mathrm{rank}(\Omega)$ may be much larger than $K$. For this reason, (\ref{logit-DCBM2}) 
is not a low-rank model in general.

{\bf Remark 1}.  Since $N_{ij} \approx 1$ when $\theta_i \theta_j \pi_i' P \pi_j \approx 0$, one may think that the  DCBM and logit-DCBM  are close to each other. This is not true. First, $\theta_i \theta_j \pi_i' P \pi_j$ are not necessarily small for all $i,j$.  Second,  even if $\widetilde{\Omega}$ and $N \circ \widetilde{\Omega}$ 
are close in each entry,  their spectra and norms can be very different. 

 
{\bf Literature review and our contribution}.  The logit-DCBM (and all other models mentioned above) are so-called latent variable models, 
where $\Pi$ is the matrix of latent variables. 
For latent variable models, the EM algorithm (e.g., \cite{EM}) is a well-known approach. 
However, EM algorithm is computationally expensive,  lacks of theoretical guarantee  
for high dimensional setting as we have here, and does not perform well when the 
networks are sparse.   For network data, penalization approach is popular, 
and in the DCBM setting, there are many interesting works (e.g., \cite{chen2018convexified,  zhao2012consistency}). 
However,  since the DCBM is a latent variable model with many unknown parameters, these methods usually involve a non-convex optimization, where a good initialization is crucial. Also, penalization approaches are usually computationally relatively slow and hard to analyze. We can extend these approaches to LSM \citep{Malatent} and logit-DCBM, but due to the nonlinearity in LSM and logit-DCBM, these issues persist. 
 
For these reasons, spectral approaches for network data are especially appealing. 
Compared with EM algorithm and penalization approaches, spectral approaches  
are conceptually simpler, computational faster, and also easier (at least for the DCBM) to analyze.    
In the classical spectral approach, we cluster by applying $k$-means to the $n$ rows of the 
matrix $\widehat{\Xi} = [\hat{\xi}_1, \ldots, \hat{\xi}_K]$, where $\hat{\xi}_k$ is the $k$-th eigenvector of $A$. 
However, due to frequently observed phenomenon of {\it severe degree heterogeneity} in network data, 
such an approach frequently performs poorly. To fix the problem, 
\citep{SCORE} proposed SCORE as a new spectral approach. In the DCBM setting, 
 SCORE was shown to have fast convergence rates (e.g., \cite{SCORE,SCORE+}). Also,  in a survey paper \citep{SCORE-Review},  SCORE was compared with many  
algorithms on many real networks, where it was shown to be competitive in real data performances.

Motivated by these, we wish to extend SCORE to our setting. The challenge is,   
the success of SCORE critically depends on the fact that the DCBM is a  low-rank model, 
but unfortunately, the logit-DCBM   is not a low-rank model (see  above).

%
%
%

We adapt SCORE to our setting by proposing the {\it Recursive-SCORE (R-SCORE)}: 
we initialize by a possibly crude estimate for $\Pi$ (denoted for $\widehat{\Pi}$), 
and then use $A$ and $\widehat{\Pi}$ to estimate $N$ (denoted by $\widehat{N}$). 
We then update $A$ by $A \oslash \widehat{N}$ ($\oslash$ denotes the entry-wise division) and repeat the above process for a number of times. The main idea here is that, if $\widehat{N} \approx N$, then $A \oslash \widehat{N} \approx A \oslash N = \widetilde{\Omega} - \diag(\widetilde{\Omega}) + W \oslash N$, where the RHS is a low-rank model (recall that 
$\mathrm{rank}(\widetilde{\Omega}) = K$).  

%
%

The challenge is, how to estimate $N$ is a {\it difficult problem}, even if $\Pi$ is known. 
In fact, when $\Pi$ is known, we can restrict the network to each of the $K$ communities, 
where within each community, the logit-DCBM reduces to the $\beta$-model (which is a symmetrical version of the $p_1$ model \citep{p1model}).    How to estimate $N$ in the $\beta$-model is a well-known open problem, as  explained in the  
survey paper \citep{goldenberg2010survey} (see also \cite{Rinaldo2010mle}):  ``A major problem with the $p_1$ and related models, recognized by Holland and Leinhardt, is the lack of standard asymptotics, ..., we have no consistency in results for the maximum likelihood estimates".

We tackle this with a cancellation trick. 
Construct two types of cycles. For each type, the expected 
cycle count  is a big sum of many terms, where due 
to  the matrix of nonlinear factors $N$,  we can not derive a simple   
expression.  Fortunately, 
in the ratio of the two big sums,  {\it the nonlinear factors in one big sum 
cancel with those in the other, and  the ratio has a simple 
and closed-form expression}.

Therefore, if $\Pi$ is known, then the idea gives rises to a simple and convenient way to estimate $N$. 
Note that this also solves the open problem for the $\beta$-model aforementioned.   
In our case,  $\Pi$ is unknown, but we can first obtain a possibly crude estimate $\widehat{\Pi}$, and 
then use $\widehat{\Pi}$ and the idea above to obtain an estimate $\widehat{N}$ for $N$. We can then 
repeat the two steps as in the R-SCORE.   

{\bf Remark 2}.  As many recent procedures rely on a low-rank network model, the above idea 
is not only useful for adapting SCORE to our setting, but is also helpful 
in adapting other ideas (e.g., those on global testing \citep{JKL2021} and on estimating $K$ \citep{JKLW2022}) 
to our setting. 
 
It remains to analyze SCORE and R-SCORE for the logit-DCBM model.  Note that while 
the Hamming clustering error of SCORE was analyzed before (e.g., \cite{SCORE, SCORE+}), 
but the focus were on the simpler DCBM model, where the analysis critically depends on that the DCBM is a low-rank model. 
Unfortunately, the logit-DCBM is not a low-rank model, so it is unclear how to extend the results in  \cite{SCORE, SCORE+} to our setting. Note also that R-SCORE is a new algorithm, and has never been analyzed before.

For any community detection procedures, we measure the performance by the Hamming clustering error. 
In the logit-DCBM model, we can always write 
\[
A = \widetilde{\Omega}  + (N - {\bf 1}_n {\bf 1}_n') \circ \widetilde{\Omega} - \diag(\Omega) + W,  \qquad \mbox{where $\widetilde{\Omega}$ is a low-rank matrix}, 
\] 
and $(N - {\bf 1}_n {\bf 1}_n') \circ \widetilde{\Omega}$ can be viewed as a non-linear perturbation of $\widetilde{\Omega}$.  
We show that the Hamming error rate of SCORE is upper bounded by 
\[
C [\lambda_1(\widetilde{\Omega})+ \|(N - {\bf 1}_n {\bf 1}_n') \circ \widetilde{\Omega}\|] / \lambda_K^2(\widetilde{\Omega}). 
\] 
This is the first time we derive a bound for the Hamming clustering error of  SCORE for a nonlinear network model. Compared with 
existing works on DCBM (e.g., \cite{JKL2021, JKLW2022}), the analysis is  quite different.

The Hamming error of R-SCORE is much harder to analyze, for many reasons. 
First, R-SCORE is a recursive algorithm, where the next step depends on 
the pervious one. Second, $\widehat{N}$ (the estimate for $N$) is a 
complicated nonlinear function of $A$, which  depends   not only on the 
clustering errors in the previous step, but also on the cycle 
count step aforementioned (where the analysis is non-standard). 

Fortunately, we manage to derive an upper bound for the Hamming error rate of R-SCORE.   To save space, we consider a special case here, 
leaving more general cases to Section \ref{subsec:theory}. Consider the special case 
where for all $1 \leq i \leq n$,  $c_0 n^{-\beta}  \leq  \theta_i \leq c_1 n^{-\beta}$, where $\beta \in (0,1/2)$ and $c_1 > c_0 > 0$ are constants. In this case, the Hamming error rates of SCORE are 
R-SCORE are upper bounded by $C n^{-a_0(\beta)}$ and $n^{-a_1(\beta)}$, respectively, 
where 
\[
a_0(\beta) = \min\{(1 - 2 \beta), 4 \beta\},  \;\;   a_1(\beta) = \min\{(1 - 2 \beta), 6 \beta\}, \;\;  \mbox{and $a_1(\beta) > a_0(\beta)$ when $\beta < 1/6$}. 
\] 
Therefore, when $0 < \beta < 1/6$, the rate of R-SCORE is faster than that of SCORE, and two rates are the same when 
$1/6 < \beta < 1/2$ (the interesting range of $\beta$ is $0 < \beta < 1/2$; when $\beta > 1/2$, the signal-noise ratio is 
so low that no procedure could succeed).   
  
We have the following contributions:   (a)  propose the logit-DCBM as an extension of the LSM, $\beta$-model, and DCBM, and as a more appealing network model,  (b) introduce a cancellation trick and use it to solve an open problem for the $\beta$-model, as well as for an open problem for the logit-DCBM in the idealized case where $\Pi$ is known, and (c) propose R-SCORE as recursive approach to community detection with the logit-DCBM, 
(d) for the first time, we derive upper bounds for the Hamming error rates of SCORE and R-SOCRE 
for the logit-DCBM (which is a nonlinear network model),   and (e) show that the rate of R-SCORE is faster than that of SCORE in 
a specific sparse region, and is at least as fast outside the region.  

In summary,  we propose the logit-DCBM as a nonlinear variant of DCBM, so it will be
more broadly accepted. The nonlinear factors make the logit-DCBM harder to fit,
but with a cancellation trick, we can successfully convert the model back to DCBM approximately,
so we can continue to enjoy all nice properties the DCBM has.  We also propose R-SCORE as a fast spectral approach
where the error rate is faster than that of applying SCORE directly to the logit-DCBM.

{\bf Content and notation}. Section \ref{sec:idea} introduces the cancellation trick. Section \ref{sec:main} 
introduces the R-SCORE algorithm and theoretical analysis. Section \ref{sec:simul} contains some numerical 
study. Section \ref{sec:Discu} discusses connections to other problems. 
In this paper, $\circ$ and $\oslash$ denote the entry-wise product and division, respectively.  For $1 \leq k \leq K$, we use $e_k$ to denote the $k$-th standard  basis vector of $\mathbb R^{K}$. For any $n \geq 2$,  $I_n$ denotes the $n \times n$ identity matrix and ${\bf 1}_n \in \mathbb R^{n}$ denotes the vector of all ones.    For any two sequence of non-negative numbers $\{a_n\}$ and $\{b_n\}$, we write $a_n\gg b_n$  if $b_n/a_n = o(1)$ (similar for $a_n \ll b_n$),  and we write $a_n\asymp b_n$ if $c_0 b_n\leq a_n\leq c_1 b_n$ for some constants $c_1 > c_0 > 0$. 
 We use $C$ to  stand for a generic constant, which may vary from one occasion to another.

\section{An idea for cancelling nonlinear terms in big sums} \label{sec:idea}  
We introduce the cancellation trick by considering two seemingly new problems.  Although it seems a digression from our original purposes, the two problems are interesting in their own right, and provide the foundation for the refitting step of R-SCORE below.    
Consider the first problem. Suppose we have a matrix $A \in \mathbb{R}^{n_1, n_2}$ with independent Bernoulli entries, 
where $\Omega_{ij} = \mathbb{P}(A_{ij} = 1) = x_0 N_{ij}   \theta_i \theta_j$,  $x_0 > 0$, $\theta_i > 0$, with $N_{ij} = [1 + x_0 \theta_i \theta_j]^{-1}$.  Here,  $\theta_i$ are known but $x_0$ is not, and the interest is to estimate $x_0$. We may estimate $x_0$ by the maximum likelihood estimate (MLE), but it does not have a closed form and may be computationally slow, so we desire a new approach. 
\begin{lemma}  \label{lemma:P}  
We have $(I) = x_0  (II)$, where $(I) =  \sum_{i, j}    \Omega_{ij}$ and   $(II) =  \sum_{i, j}   \theta_i \theta_j (1 - \Omega_{ij})$. 
\end{lemma} 
{\bf Proof}. As  $(I) = x_0  \sum_{i, j} N_{ij} \theta_i \theta_j$ and $(II) = 
\sum_{i, j}   N_{ij} \theta_i \theta_j$, the claim follows.   $\square$ 

The key  is, due to the non-linear terms $N_{ij}$, it is hard  to derive 
a closed-form formula for $(I)$ or $(II)$, but 
by our careful design,   the ratio  of $(I) / (II)$   has a very simple form. 
Now, to estimate $x_0$, let  
$\psi_n^{(1)} =  \sum_{i=1}^{n_1} \sum_{j = 1}^{n_2}   A_{ij}$ and $\psi_n^{(2)} =   \sum_{i=1}^{n_1} \sum_{j = 1}^{n_2}   \theta_i \theta_j (1 - A_{ij})$. 
By Lemma \ref{lemma:P},  
$x_0 = \mathbb{E}[\psi_n^{(1)}] / \mathbb{E}[\psi_n^{(2)}]$, so a convenient estimate for $x_0$ is (note: the computational cost is $O(n_1 n_2)$): 
\begin{equation}\label{x0estimate} 
\hat{x}_0 = \psi_n^{(1)} / \psi_n^{(2)} =  [\sum_{i=1}^{n_1} \sum_{j = 1}^{n_2}   A_{ij}] /  [\sum_{i=1}^{n_1} \sum_{j = 1}^{n_2}   \theta_i \theta_j (1 - A_{ij})]. 
\end{equation} 

\vspace{-1.5 em} 

Consider the second problem. Suppose we have a network adjacency matrix $A \in \mathbb{R}^{n_1, n_2}$ satisfying the $\beta$-model.    That is,  the upper triangle of $A$ are independent Bernoulli satisfying $\Omega_{ij} \equiv  \mathbb{P}(A_{ij} = 1) = N_{ij}  \theta_i \theta_j$ where  $N_{ij} = [1 + \theta_i \theta_j]^{-1}$, $1 \leq i \neq j \leq n_1$. The parameters $\theta_i > 0$ are unknown and the interest is to estimate them.  
Due to the nonlinear terms $N_{ij}$, the problem remains a difficult open problem in the literature, where  
classical approaches such as the MLE face grand challenges (e.g., \cite{goldenberg2010survey, Karwa:Slakovic:2016, Rinaldo2010mle}).  We propose a new approach, motivated by the following lemma. For any $1 \leq i \leq n_1$, let $S_i = \{1, 2, \ldots, n_1\} \setminus \{i\}$. 
\begin{lemma}  \label{lemma:theta} 
Fix an odd number $m \geq 3$.    We have (dist below stands for distinct) 
\begin{equation} \label{theta10} 
\frac{\sum_{i_2, \ldots, i_m \in S_{i_1} (dist)} \Omega_{i_1 i_2}(1 - \Omega_{i_2 i_3}) \ldots 
\Omega_{i_{m-2} i_{m-1}} (1 - \Omega_{i_{m-1} i_m})  \Omega_{i_m i_1}}{\sum_{i_2, \ldots, i_m \in S_{i_1} (dist)} (1 - \Omega_{i_1 i_2})  \Omega_{i_2 i_3}  \ldots 
(1 - \Omega_{i_{m-2} i_{m-1}})  \Omega_{i_{m-1} i_m} (1 - \Omega_{i_m i_1})}  = \theta_{i_1}^2. 
\end{equation} 
\end{lemma} 
{\bf Proof}.  Note that 
%
\begin{align*} 
\Omega_{i_1 i_2}(1 - \Omega_{i_2 i_3}) \ldots 
\Omega_{i_{m-2} i_{m-1}} (1 - \Omega_{i_{m-1} i_m})  \Omega_{i_m i_1} &  = N_{i_1 i_2} N_{i_2 i_3} \ldots 
N_{i_m i_1} \theta_{i_1}^2 \theta_{i_2} \ldots \theta_{i_m} \\ 
(1 - \Omega_{i_1 i_2})  \Omega_{i_2 i_3}  \ldots 
(1 - \Omega_{i_{m-2} i_{m-1}})  \Omega_{i_{m-1} i_m}   (1 - \Omega_{i_m i_1})  & = N_{i_1 i_2} N_{i_2 i_3} \ldots 
N_{i_m i_1}   \theta_{i_2} \ldots \theta_{i_m}.  
\end{align*} 
Comparing the RHS, the only difference is the term $\theta_i^2$. Since on both the numerator and denominator of (\ref{theta10}), the sum is only over $i_2, i_3, \ldots, i_m$ with $i_1$ being fixed, the claim follows.  $\square$

Similarly,  due to the non-linear terms $N_{ij}$, it is hard to derive 
a closed-form formula for  both the numerator and denominator of (\ref{theta10}), but 
by our design,  the ratio in (\ref{theta10})  has a very simple form. 
Let $\phi_{n,m}^{(1)}(i_1)  = \sum_{i_2, \ldots, i_m \in S_{i_1} (dist)} A_{i_1 i_2}(1 - A_{i_2 i_3}) \ldots  A_{i_{m-2} i_{m-1}} (1 - A_{i_{m-1} i_m}) A_{i_m i_1}$ and     
$\phi_{n,m}^{(2)}(i_1) =  \sum_{i_2, \ldots, i_m \in S_{k, i_1} (dist)}   (1 - A_{i_1 i_2}) A_{i_2 i_3}  \ldots (1 - A_{i_{m-2} i_{m-1}})   A_{i_{m-1} i_m} (1 -  A_{i_m i_1})$.   
By Lemma \ref{lemma:theta},  
\[
\frac{\mathbb{E}[\phi_{n, m}^{(1)}(i_1)]}{\mathbb{E}[\phi_{n, m}^{(2)}(i_1)]}  = \frac{\sum_{i_2, \ldots, i_m \in S_{i_1}} \Omega_{i_1 i_2}(1 - \Omega_{i_2 i_3}) \ldots 
\Omega_{i_{m-2} i_{m-1}} (1 - \Omega_{i_{m-1} i_m})  \Omega_{i_m i_1}}{\sum_{i_2, \ldots, i_m \in S_{i_1}} (1 - \Omega_{i_1 i_2})  \Omega_{i_2 i_3}  \ldots 
(1 - \Omega_{i_{m-2} i_{m-1}})  \Omega_{i_{m-1} i_m} (1 - \Omega_{i_m i_1})}  = \theta_{i_1}^2.
\] 
Therefore, a reasonable estimator for $\theta_{i_1}$ is 
$\hat{\theta}_{i_1} = \sqrt{\phi_{n, m}^{(1)}(i_i)  / \phi_{n, m}^{(2)}(i_1)}$. This solves the open problem aforementioned (see also Section \ref{sec:main}).  
Especially, we may take $m = 3$  and  estimate $\theta_i$ by  
\begin{equation} \label{thetaestimate} 
\hat{\theta}_{i} = \sqrt{\phi_{n, 3}^{(1)}(i)  / \phi_{n,3}^{(2)}(i)} = \sqrt{\frac{\sum_{j, k  \in S_{i}, j \neq k} A_{ij}(1 - A_{jk}) A_{ki}}{\sum_{j, k \in S_i, j \neq k} (1 - A_{ij})  A_{jk} (1 - A_{ki})}}. 
\end{equation} 
 Alternatively, we may use a larger $m$, but the numerical performance is similar, while the analysis is much longer. For each fixed $m$, the computational cost   is $O(n^2 d)$ (e.g., \cite{JKL2021}), where $d$ is the maximum node degree.

{\bf Remark 3}. Lemma \ref{lemma:theta} is readily extendable. For example, if there are positive functions $g$ and $h$ such that $g(\Omega_{ij}) = \theta_i \theta_j  \pi_i' P \pi_j h(\Omega_{ij})$ for all $i, j$,   then similarly 
\[ 
\frac{\sum_{i_2, \ldots, i_m \in S_{i_1} (dist)} g(\Omega_{i_1 i_2}) h(\Omega_{i_2 i_3}) \ldots 
g(\Omega_{i_{m-2} i_{m-1}}) h(\Omega_{i_{m-1} i_m})  g(\Omega_{i_m i_1})}{\sum_{i_2, \ldots, i_m \in S_{i_1} (dist)}  h(\Omega_{i_1 i_2})  g(\Omega_{i_2 i_3})  \ldots 
h(\Omega_{i_{m-2} i_{m-1}})  g(\Omega_{i_{m-1} i_m}) h(\Omega_{i_m i_1})})  = \theta_{i_1}^2. 
\] 

In summary, the two problems above (especially the second one) are difficult. In these problems, the quantities of interest are hidden in some big sums. Due to the nonlinear factor $N_{ij}$ , it is hard to derive a closed-form formula for such big sums. However, if we can carefully construct two big sums, then we can cancel the nonlinear terms $N_{ij}$ by considering the ratio of the two big sums, and derive a closed-form formula for the quantity of interest. Such a cancellation trick gives rises to a convenient approach to solving the two problems above, and is readily extendable to many other settings (e.g., analysis of the p1 model for directed networks \cite{p1model}, analysis of tensor and hyper-graphs \cite{Yuan2018TestHyper}).

Below in Section \ref{sec:main}, we introduce R-SCORE as a recursive algorithm, where the ideas above play a key role in the refitting steps of R-SCORE. 
For space reasons, we defer the analysis of the above idea (i.e., $\hat{x}_0$ and $\hat{\theta}_i$) to 
the supplement; see Sections C.2-C.3 of the supplement for details. 
%
%

%
%
%
%

\section{Community detection by R-SCORE for the logit-DCBM}  \label{sec:main}  
We propose {\it Recursive-SCORE (R-SCORE}) for community detection, 
where the key is to use the ideas above in the refitting step; see Algorithm (\ref{alg:RSCORE}).   
The number of iteration is not critical, so we set $M = 10$ (R-SCORE typically converges in very few iterations).
In each iteration, R-SCORE consists a community detection step by SCORE (the SCORE step)  and a refitting step. 
We choose SCORE for it is fast, competitive in real data analysis, and with fast error rates (e.g.,  \cite{SCORE, SCORE+}), but we can also view our algorithm as a {\it generic algorithm},   where we can replace the SCORE by any other community detection approaches that are provably effective for DCBM.

We now discuss the SCORE step and refitting step of Algorithm \ref{alg:RSCORE} in detail.  Consider the SCORE step  \citep{SCORE} first. In this step, for an input matrix $A$ or $A \oslash \widehat{N}$,   
let $\hat{\xi}_1, \ldots, \hat{\xi}_K$ be the first $K$ eigenvectors,  and let $\widehat{R} = [\hat{\xi}_2/ \hat{\xi}_1, \ldots,\hat{\xi}_K / \hat{\xi}_1]$, where $\xi / \eta$ denotes the vector of entry-wise ratios. We cluster by applying the $k$-means to the $n$ rows of $\widehat{R}$, and let $\hat{\pi}_i$ be the estimated community label of node $i$. Let $\widehat{\Pi} = [\hat{\pi}_1, \ldots, \hat{\pi}_n]'$.  
Note that $\hat{\pi}_i$ takes values in $e_1, e_2, \ldots, e_K$ ($e_k$: $k$-th 
standard basis vector of $\mathbb{R}^K$).  
%
%
%
 
\begin{algorithm}[htb!]
\caption{The Recursive SCORE (R-SCORE)}  \label{alg:RSCORE}
{\bf Input}: $A$ and $K$.
Initialize with an estimate $\widehat{\Pi}$ by SCORE.  For $m=1,2,\ldots,M$, 
\begin{itemize}\itemsep 0pt
\item {\it Refitting}. Update $\widehat{N}$ using $A$, $\widehat{\Pi}$ in the most recent step, and the refitting step below.  
\item {\it SCORE}.  Update $\widehat{\Pi}$ by applying SCORE to $A \oslash \widehat{N}$ with the most recent $\widehat{N}$.  
\end{itemize}
{\bf Output}: $\widehat{\Pi} = [\hat{\pi}_1, \ldots, \hat{\pi}_n]'$.  
\end{algorithm}

%
%

Consider the refitting step. Let $\widehat{\Pi} = [\hat{\pi}_1,\ldots, \hat{\pi}_n]'$ be the estimated $\Pi$ in the current iteration. Recall that even in the idealized case of $\widehat{\Pi} = \Pi$, refitting (i.e., how to estimate $N$)  is a difficult and open problem. We tackle this with the idea in Section \ref{sec:idea}. 
In detail, let ${\cal C}_k = \{1 \leq i \leq n: \hat{\pi}_i = e_k\}$ be the $k$-th estimated 
community, $1 \leq k \leq K$, and let $\hat{S}_{k, i} = \widehat{{\cal C}}_k  \setminus \{i\}$.  By Lemma \ref{lemma:theta} and especially (\ref{thetaestimate}),   we propose to estimate $\theta_i$ by 
\begin{equation} \label{theta2} 
\hat{\theta}_i = \sqrt{(\sum_{j \neq k \in \hat{S}_{k,i}}  A_{ij} (1 - A_{jk})    A_{ki}) / (\sum_{j \neq k \in \hat{S}_{k,i}}  (1 - A_{ij})  A_{jk} (1 - A_{ki}))}, \qquad \mbox{if $i \in \widehat{{\cal C}}_k$}. 
\end{equation} 
This corresponds to the case of $m = 3$ of our idea in Section \ref{sec:idea}, but we can also use a larger $m$. 
Also, inspired by Lemma \ref{lemma:P}, we can estimate the matrix $P$ by 
\begin{equation} \label{P2} 
\widehat{P}_{k\ell} =  [\sum_{i \in \widehat{{\cal C}}_k}  \sum_{j  \in {\widehat{\cal C}}_{\ell}}  A_{ij}] /  [\sum_{i \in \widehat{{\cal C}}_k}  \sum_{j  \in {\widehat {\cal C}}_{\ell}}\hat{\theta}_i \hat{\theta}_j (1 - A_{ij})], \qquad 1 \leq k, \ell \leq K.  
\end{equation} 
To appreciate the idea, consider the sub-matrix of $A$ by restricting the rows and columns to $\widehat{{\cal C}}_k$ and $\widehat{{\cal C}}_{\ell}$. In the idealized case where $\hat{\theta}_i = \theta_i$ and $\widehat{{\cal C}}_k = {\cal C}_k$,  $1 \leq i \leq n$,  $1 \leq k \leq K$,  the mean of the sub-matrix satisfies the condition of Lemma \ref{lemma:P} of Section \ref{sec:idea}.  This gives rises to the estimates above. 
Finally, we update our estimate of $N$ by letting 
$\widehat{N}_{ij} = (1+\hat{\theta}_i \hat{\theta}_j \widehat{P}_{k \ell})^{-1}$ if $i \in \widehat{{\cal C}}_k$ and $j \in \widehat{{\cal C}}_{\ell}$.  

Note that by the discussion in Section \ref{sec:idea}, the computational cost of the refitting step is no more than $O(n^2 d)$, where $d$ is the maximum node degree. As a result, the computational cost of R-SCORE is no more than $O(n^2 d)$. 

{\bf Remark 4}. One may want to replace the refitting step by a simpler step, but 
how to do so remains unclear. Recall that even when $\Pi$ is known, how to 
estimate $(\Theta, P)$ is a challenging problem (e.g., \cite{goldenberg2010survey, Karwa:Slakovic:2016, Rinaldo2010mle}).

\subsection{Theoretical results} \label{subsec:theory} 
 For any community detection procedure, let $\widehat{\Pi}$ be the resultant estimate for $\Pi$, where each row of $\widehat{\Pi}$ takes values 
in $\{e_1, e_2, \ldots, e_K\}$.   We measure the performance by the Hamming error rate: 
\[
r_n(\widehat{\Pi})   = \frac{1}{n}  \min_{{\cal P}} \sum_{i = 1}^n 1\{\hat{\pi}_i \neq {\cal P} \pi_i\}, \qquad 
\mbox{where ${\cal P}$ is any permutation in $\{1, 2, \ldots, K\}$}.  
\] 
Let $\widehat{\Pi}^{score}$ and $\widehat{\Pi}^{rscore}$ be the $\widehat{\Pi}$ for 
applying SCORE and R-SCORE to the adjacency matrix $A$.   

Note that under the logit-DCBM model,  
\begin{equation} \label{logit-DCBM3} 
A = \Omega - \diag(\Omega) + W, \;  \mbox{where $\Omega = N  \circ \widetilde{\Omega}$, \;  $\widetilde{\Omega} =  \Theta \Pi P \Pi' \Theta$ and $P$ has unit diagonal entries}.  
\end{equation} 
The last item is a well-known {\it identifiability condition} \citep{JKLW2022}.   Since in most real networks,  $K$ is relatively small, so we suppose that $K$ is fixed (this is only for technical simplicity and can be relaxed). 
Let $n_k$ be the number of nodes in the $k$-th community,  $1 \leq k \leq K$. We assume 
\begin{align}\label{asm:Pi}
\min_{k} \{n_k\} \geq c_0 n, \quad \text{ for some constant $0<c_0\leq 1/K$}.  
\end{align}
This is a frequently used and  {\it mild balance condition among the $K$ communities} (e.g., \cite{SCORE, JKL2021}).  Also, we assume  that there exists constants $c_2\geq c_1>0$ such that 
\begin{align}\label{asm:Theta}
\bar{\theta} \goto 0, \qquad \mbox{and} \qquad  \mbox{$c_1\bar \theta \leq \theta_i \leq c_2 \bar \theta$ for all $1\leq i \leq n$},  \qquad \mbox{where $\bar \theta = \sum_{i=1}^n \theta_i/n$}.  
\end{align} 
This condition is also 
only for technical reasons, and can be largely relaxed. 
Furthermore, we assume  there exists an constant $c_3>0$, such that
\begin{align}\label{asm:P}
 \sqrt{ n} \bar \theta \cdot | \lambda_{\min}(P)| \geq c_3 \log(n), \qquad \mbox{$\lambda_{min}(P)$:  smallest eigenvalue of $P$ in magnitude}.  
\end{align}
In the special case where $A$ satisfies a DCBM, $\Omega = \widetilde{\Omega}$, and SCORE 
was analyzed before (e.g.,  \cite{SCORE, SCORE+}), where it is known that the signal-noise ratio (SNR) is 
given by $|\lambda_K(\widetilde{\Omega})| / \lambda_1^{1/2}(\widetilde{\Omega})$ ($|\lambda_K(\widetilde{\Omega})|$ and $\lambda_1^{1/2}(\widetilde{\Omega})$ represent the 
signal and noise level respectively). In order for the Hamming error rate of SCORE tends to $0$, it is necessary that the 
SNR $\goto \infty$. Condition (\ref{asm:P}) is necessary for otherwise the SNR may tend to $0$. 
%
%
%
Also,  here $\lambda_{\min}(P)$ measures community dissimilarity. In the special case of $P = b{\bf 1}_K {\bf 1}_K + (1- b) I_K$,   $0<b<1$,   $\lambda_{\min}(P) = 1-b$. Therefore, if $\lambda_{\min}(P) \to 0$, then  $b \to 1$, and all $K$ communities are very similar. Condition (\ref{asm:P}) defines a class of {\it weak signal settings} where the problem of community detection is challenging.  
Lastly, consider $P \Pi '\Theta^2 \Pi$. Let $\eta$ be the first right eigenvector of $P \Pi '\Theta^2 \Pi$, we assume that $\eta$ is a positive vector and 
\begin{equation} \label{asm:eta} 
\lambda_1(P \Pi '\Theta^2 \Pi) -| \lambda_2(P \Pi '\Theta^2 \Pi)|\geq c_4 \lambda_1(P \Pi '\Theta^2 \Pi), \qquad  \max_{i}\eta(i)/\min_{i}\eta(i)  \leq c
\end{equation} 
This condition is necessary to guarantee that the first eigenvector is well-separated from the others and  the SCORE normalization by the first eigenvector is well-defined, since the reciprocal of each entry of the first eigenvector cannot blow up. It is a mild condition by Perron's theorem on non-negative matrices. Similar condition can be found in \cite{JKLW2022}.

Note that while SCORE was analyzed before for the DCBM, 
it was not analyzed for the logit-DCBM,  where the analysis 
is  expected to be much harder. 
In the logit-DCBM,  we have $\Omega = \widetilde{\Omega} + (N - {\bf 1} {\bf 1}') \circ \widetilde{\Omega}$. 
To avoid that the nonlinearity completely ruins the low-rank structure,  we need 
\begin{equation} \label{asm:score-main} 
\|(N - {\bf 1}_n{\bf 1}_n') \circ \widetilde{\Omega})\| /| \lambda_K(\widetilde{\Omega})| \to 0 ; 
\end{equation} 
Recall that ${\rm SNR} = |\lambda_K(\widetilde{\Omega})| / \lambda_1^{1/2}(\widetilde{\Omega})$. The following theorem 
is proved in the supplement.
%
\begin{theorem} \label{thm:SCORE}
Let $\widehat{\Pi}^{score}$ be the resultant estimate for $\Pi$ when we apply SCORE directly to $A$ and suppose (\ref{asm:Pi})-(\ref{asm:eta}) and (\ref{asm:score-main}) hold. 
With probability $1- o(n^{-3})$,
\[
r_n(\widehat{\Pi}^{score}) \leq  C [\| (N-{\bf 1}_n{\bf 1}_n')  \circ \widetilde \Omega\|^2  + \lambda_1(\widetilde \Omega)] / \lambda_K^2(\widetilde \Omega). 
\]
In the special case where $A$ satisfies the DCBM, $N = {\bf 1}_n{\bf 1}_n'$, and 
$r_n(\widehat{\Pi}^{score})  \leq C \lambda_1(\widetilde{\Omega}) / \lambda_K^2(\widetilde{\Omega})$. 
\end{theorem} 

Next, we consider R-SCORE. Since R-SCORE is a recursive algorithm, it is useful to present 
a result that is {\it applicable in general cases}. Consider an estimate for $\widetilde \Omega$ in the form of  $\widehat{\widetilde \Omega} = \widehat \Theta \widehat \Pi \widehat P \widehat \Pi' \widehat \Theta$. By our construction, $\widehat{N}_{ij} = 1 /(1 + \widehat{\widetilde{\Omega}}_{ij})$.  Suppose that with probability $1 - o(n^{-3})$, 
\begin{equation} \label{asm:rscore1a} 
\|\widehat P - P\|_{\max}\ll \min\{ 1, |\lambda_{\min}(P) |\bar \theta^{-1}\}, \quad  \|\widehat \Pi - \Pi\| (\sqrt{n}\, |\lambda_{\min}(P)|)^{-1} \bar\theta \to 0, 
\end{equation} 
and 
\begin{equation} \label{asm:rscore1b} 
\|(N \oslash \widehat N  - {\bf 1}_n{\bf 1}_n') \circ \widetilde \Omega\|  = o(|\lambda_{K} ( \widetilde \Omega)| ), \qquad \|(N \oslash \widehat N  - {\bf 1}_n{\bf 1}_n') \|_{F} = o( \lambda_1(\widetilde \Omega) )\,. 
\end{equation} 
The following lemma is proved in the supplement. 
\begin{lemma}\label{lemma:RSCORE}
Suppose (\ref{asm:Pi})-(\ref{asm:eta}) hold.  Let $\widehat{\Pi}$ be the 
result of applying SCORE to $A \oslash \widehat{N}$ where (\ref{asm:rscore1a})-(\ref{asm:rscore1b}) and $\hat \theta_i <C \bar \theta$ hold. With probability $1- o(n^{-3})$, 
\[
r_n(\widehat{\Pi})  \leq   \frac{C[\| (N\oslash \widehat N - {\bf 1}_n{\bf 1}_n')  \circ \widetilde \Omega\|^2  +\tau_n^2+   { \lambda_1(\widetilde \Omega)}]}{\lambda_K^2(\widetilde \Omega) },   \mbox{where 
$\tau_n =   \sqrt{n}\bar \theta^3 [\sqrt{n} \| \widehat P - P\|_{\max}  +   \| \widehat \Pi - \Pi\|]$}.
\] 
\end{lemma}
 
We now show that (\ref{asm:rscore1a})-(\ref{asm:rscore1b}) hold in many settings. 
Our numerical study shows that R-SCORE typically converges in just one iteration, so for convenience in analysis, we consider R-SCORE with one iteration from now on in this section.  Write for short $\delta_n = \max\{\|(N - {\bf 1}_n {\bf 1}_n') \circ 
\widetilde{\Omega}\|^2, \lambda_1(\widetilde{\Omega})\} / \lambda_K^2(\widetilde{\Omega})$.  
Theorem \ref{thm:main}  is proved in the supplement. 
\begin{theorem}  \label{thm:main} 
Suppose (\ref{asm:Pi})-(\ref{asm:eta}) hold and  
$\delta_n / \min\{\bar \theta^2, \bar \theta |\lambda_{\min} (P)|,  |\lambda_{\min}(P)|^2 /\bar \theta^2\} \goto 0$.    
Let $\widehat{\Pi}^{rscore}$ be the estimate for $\Pi$ by applying R-SCORE to $A$, and 
let $(\widehat{\Theta}, \widehat{P}, \widehat{\Omega}, \widehat{N})$ be the corresponding 
estimates for $(\Theta, P, \Omega, N)$ in the refitting step of R-SCORE. 
We have that (\ref{asm:rscore1a})-(\ref{asm:rscore1b}) hold and that 
with probability $1 - o(n^{-3})$,  
\[
r_n(\widehat{\Pi}^{rscore})  \leq \frac{C}{\lambda_K^2(\widetilde \Omega)} \Big( \lambda_1(\widetilde{\Omega}) + n\bar \theta^4 \log(n) + n^2 \bar \theta^2 \delta_n^2 + n^2 \bar \theta^6 \delta_n\Big). 
\]
\end{theorem} 

\begin{corollary} \label{cor:main} 
Suppose (\ref{asm:Pi})-(\ref{asm:eta}) hold,  $|\lambda_{\min} (P)| \geq C$ for a constant $C > 0$,  and 
$n \bar{\theta}^4  \goto \infty$. Let $\widehat{\Pi}^{score}$ and $\widehat{\Pi}^{rscore}$ be the 
estimates for $\Pi$ by applying SCORE and R-SCORE to the adjacency matrix $A$, respectively. 
With probability $1 - o(n^{-3})$, 
\[
r_n(\widehat{\Pi}^{score})  \leq C \Big( \frac{1}{n\bar \theta^2} + \bar \theta^4\Big), \qquad r_n(\widehat{\Pi}^{rscore})  \leq C\Big( \frac{1}{n\bar \theta^2} + \bar\theta^6  +   \frac{\log(n)}{n}\Big). 
\]
\end{corollary} 
With a more careful analysis, we conjecture that the condition of $n \bar{\theta}^4$ can be 
removed, and the rate of R-SCORE is at least as fast as that of SCORE in the whole range 
of interest (note that the proof of Theorem \ref{thm:main} and Corollary \ref{cor:main} 
is already hard and relatively long).  If we calibrate $\bar{\theta} = n^{-\beta}$ for a constant $\beta > 0$, then in order for the SNR $\goto \infty$  (e.g., see  (\ref{asm:score-main})), we must have $0 < \beta < 1/2$.  In this range, 
$r_n(\Pi^{score}) \leq C n^{-a_0(\beta)}$ and $r_n(\Pi^{rscore}) \leq C n^{-a_1(\beta)}$, where 
\[
a_0(\beta) = \left\{
\begin{array}{ll} 
4\beta,   &   0 <  \beta  \leq  1/6, \\
1 - 2 \beta, &   1/6 < \beta < 1/2,   
\end{array} 
\right.  \qquad   a_1(\beta) =  \left\{
\begin{array}{ll} 
6 \beta,    &   0 <  \beta  \leq  1/8, \\ 
(1 - 2 \beta), &  1/8 < \beta \leq 1/6, \\ 
(1 - 2 \beta), &   1/6 < \beta < 1/2. 
\end{array} 
\right.; 
\] 
see Figure \ref{fig:add}. Therefore, when $0 < \beta < 1/6$, the Hamming error rate of R-SCORE is faster than that of SCORE.    
When $\beta > 1/6$, such a conclusion may also be true, as the current bound may be 
conservative: the Hamming error rate for R-SCORE depends on a complicated data dependent matrix $\widehat{N}$, 
and the error bound can hopefully be improved with more careful analysis. 

{\bf Remark 5}. The proof of Theorem \ref{thm:main} and Corollary \ref{cor:main} is hard, for 
$\widehat{N}$ has a complicated form: recall that $\widehat{N}_{ij} = 1 / (1 + \widehat{\widetilde{\Omega}}_{ij})$ and 
$\widehat{\widetilde\Omega} = \widehat \Theta \widehat \Pi \widehat P \widehat \Pi' \widehat \Theta$, where 
$\widehat{\Pi}$ is the SCORE estimate for $\Pi$, and $(\widehat{\Theta}, \widehat{P})$ are 
constructed using $\widehat{\Pi}$, $A$, and a cancellation trick. Note that even when $\Pi$ is known, 
how to estimate $(\Theta, P)$ is a nontrivial problem and we resolve it with a cancellation trick.

\begin{figure}[htbp]
\centering
\includegraphics[width=0.5\linewidth, height = 0.30\linewidth]{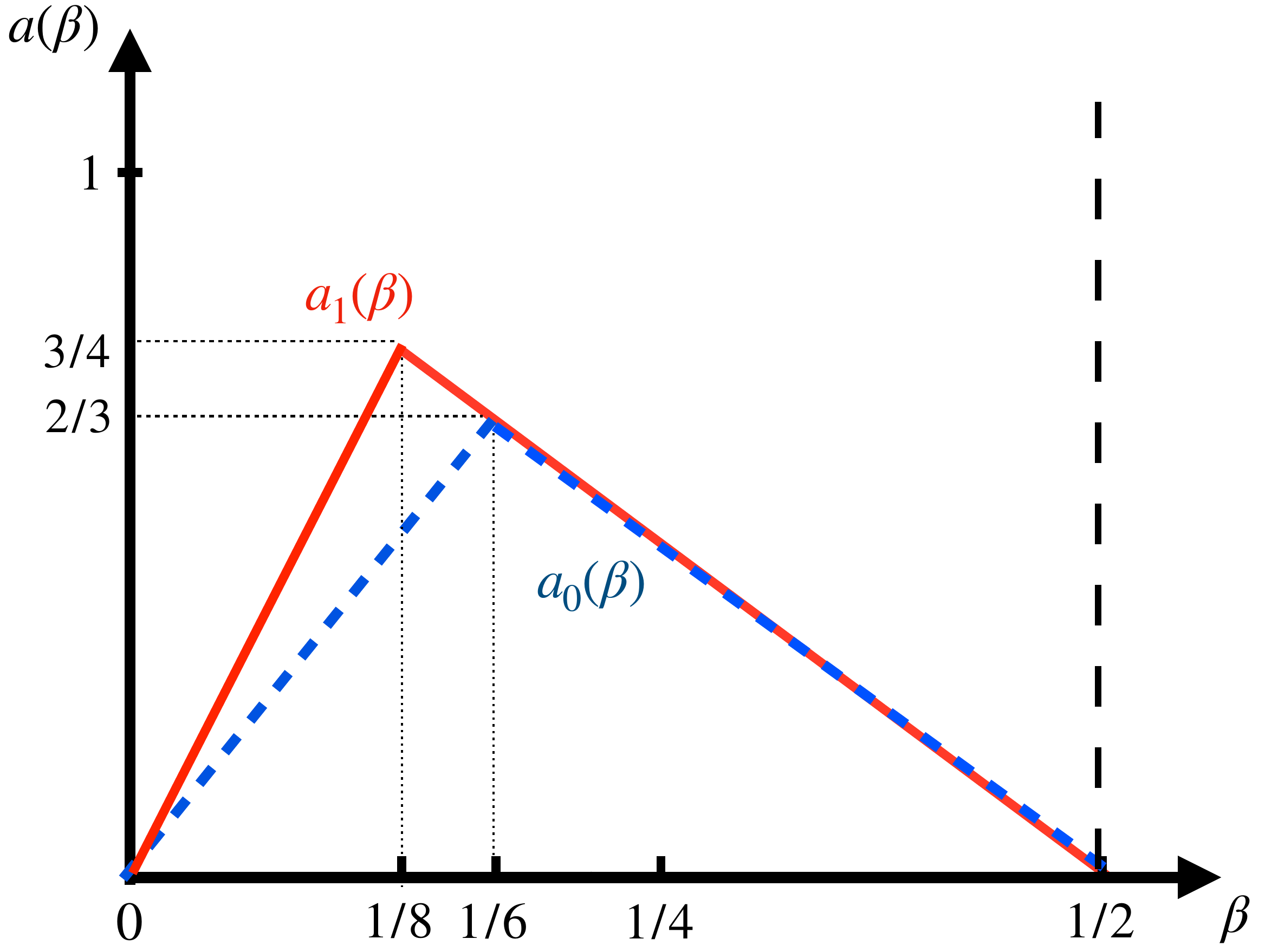}
\caption{Comparison of error rates ($x$-axis: $\beta$. $y$-axis: $a_0(\beta)$ (blue) and $a_1(\beta)$ (red)) .}
\label{fig:add} 
\end{figure}

\section{Simulation results} \label{sec:simul}  
We compare R-SCORE with SCORE and a non-convex penalization MLE-based approach 
by \citep{ma2020universal}, which we refer to as npMLE.  We compare with npMLE for \citep{ma2020universal} 
deals with the LSM and is probably the closest related work to our paper. 
Our study contains $3$ experiments.  In Experiment 1, we compare R-SCORE with SCORE (which is viewed as a benchmark).  In Experiment 2, we study how the error rates of R-SCORE and npMLE change across different iterations (both algorithms are recursive). In Experiment 3, we compare the errors of R-SCORE and npMLE. 

{\bf Experiment 1.} {\it R-SCORE vs. SCORE}.  
The networks are simulated as follows: fixing $(n, K)$, we first simulate an $n \times n$ matrix $\Omega$ as in the logit-DCBM model as follows. 
  Let $P = \beta{\bf 1}_K {\bf 1}_K + (1- \beta) I_K$ and
 $\Pi = [\pi_1, \pi_2, \ldots, \pi_n]'$, where $\pi_i = e_k$ for $n_0 = n/K$ different $i$ (recall that 
 $e_k$ is the $k$-th Euclidean basis vector of $\mathbb{R}^K$).  Moreover, fixing a parameter $b_n > 0$, we generate $\theta$ as follows.  We first draw $\theta_1^0, \theta_2^0, \ldots, \theta_n^0$ i.i.d. from a fixed distribution $F(\cdot)$, and then renormalize $\theta_0$ to get $\theta_i = b_n\cdot \theta_i/\|\theta\|$ for $1\leq i\leq n$.  
Finally, we let $\widetilde{\Omega} = \Theta \Pi P \Pi' \Theta$ and $\Omega_{ij} = \widetilde{\Omega}_{ij} / (1 + \widetilde{\Omega}_{ij})$, $1 \leq i, j \leq n$. Once we have such a matrix $\Omega$, we use it to generate a binary adjacency  matrix $A$. 
 
In such settings,   approximately,   the Signal-to-Noise ratio (SNR)  is $b_n(1- \beta)$ (e.g., see \cite{JKL2021}). 
It is desirable to choose settings that the SNR is neither too large or too small.   
Consider four settings (A), (B), (C) and (D).  In Setting (A), we fix $(n, K) = (2400, 3)$ and  $F = {\rm Uniform} (0.01, 2)$.  We choose 
$b_n = 60$ and $\beta = 23/30$ (and this way,  ${\rm SNR} = 14$). 
In Setting (B), we fix $(n, K) = (2500, 5)$ and  $F = {\rm Uniform} (0.1, 0.8)$.  We choose $b_n = 70$ and $\beta = 0.65$ (and so ${\rm SNR} = 24.5$). In Setting (C), we fix $(n, K) = (2400, 3)$ and  $F = {\rm Pareto} (10, 1)$. To avoid extremely severe degree heterogeneity, we truncate each $\theta_i^0$ at $200$.  We choose 
$b_n = 70$ and $\beta = 0.55$ (and this way,  ${\rm SNR} = 31.5$). 
In Setting (D), we fix $(n, K) = (2500, 5)$ and  $F = {\rm Pareto} (10, 1)$ with truncation at $100$.  We choose $b_n = 50$ and $\beta = 0.55$ (and so ${\rm SNR} = 22.5$).

The results are in Figure \ref{fig:Pi_error}, where the $x$-axis is the $\#$ of iterations $m$, and 
the $y$-axis is the corresponding error rate by R-SCORE (green dashed line:  error rate for 
R-SCORE with $m = 0$, same as that of applying SCORE directly to the adjacency matrix $A$).      In all settings above, the performance of directly applying SCORE ($m = 0$) is unsatisfactory. The improvements achieved by R-SCORE are significant, with substantially reduced error rates. This suggests that (a) the R-SCORE is successful as we expect, and (b) the iteration algorithm typically converges in very few iterations. Based on the numerical results, we believe that the refitting procedure steps  in R-SCORE are effective: by re-normalization, they reduce 
the logit-DCBM model approximately to a low-rank model. 
 
{\bf Experiment 2.} {\it Error rates of R-SCORE and npMLE in different iterations}. 
Fix $(n, K) = (5400, 6)$. Let $\Pi$ be generated similarly to Experiment 1 except that $\pi_ i = e_k$ for $n_k$ different $i$. Let $P = \left[\begin{array}{ll}
P_1 & P_2 \\ 
P_2 & P_1 \\  
\end{array} \right]$, where  $P_1= 0.5  \beta_1{\bf 1}_{K/2} {\bf 1}_{K/2} + (1- 0.5 \beta_1) I_{K/2}$ and $P_2 = 0.5 (\beta_1 + \beta_2){\bf 1}_{K/2} {\bf 1}_{K/2}$.  We generate $\theta_i$ in the same manner as in Experiment 1 and similarly construct $\widetilde \Omega$ and $\Omega$   to generate a binary adjacency matrix $A$.  
In this experiment, we choose $(n_1, n_2, \cdots, n_K) = 200 \cdot(5, 1.5, 6, 3, 7.5, 4)$, $F = {\rm Uniform} (0.01, 2)$, $b_n= 80$, and $(\beta_1, \beta_2)= (0.9, 0.6 )$. Under this setting, the SNR is given by $b_n |\lambda_{\min}(P)|= 28$. 
For each $m = 1, 2, \ldots, 1000$, we apply R-SCORE and npMLE with $m$ iterations (SCORE is also included for comparison, which is the same as R-SCORE with $m = 0$).  The result is in Figure~ \ref{fig:compare_nonconvex} (left).  
We observed that (a) R-SCORE converges much more rapidly than the npMLE, and (b) 
the error rates of R-SCORE is significantly lower than that of SCORE, and is slightly lower 
than that of npMLE.  
%
%

%
%

\begin{figure}[htbp]
\centering
\includegraphics[width=\linewidth, height = 0.28 \linewidth]{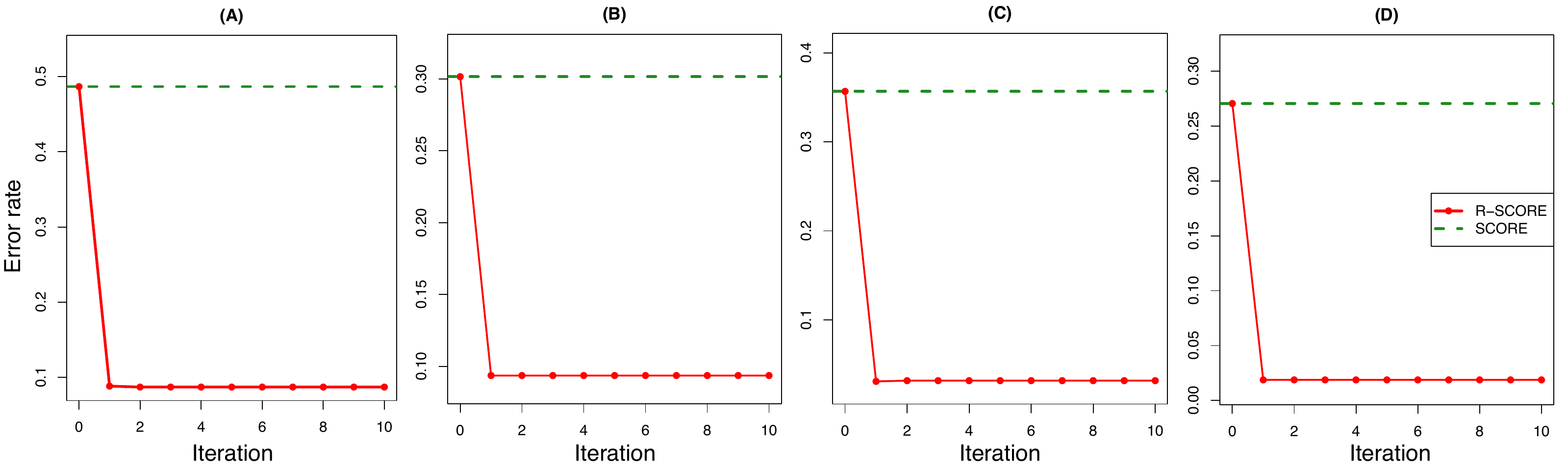}
\caption{The error rates of R-SCORE for different $m$ ($\#$ of iterations). See Experiment 1 for setting details. 
 SCORE is also included for comparision as a benchmark, which corresponding to R-SCORE with $m = 0$ ($x$-axis: $m$; $y$-axis: error rate).}
\label{fig:Pi_error} 
\end{figure}

{\bf Experiment 3.}  {\it R-SCORE vs. npMLE}.   Consider the same setting as in Experiment 2, but we let $\beta_2$ vary: we  
set $b_n = 30$ and let $\beta_2$ range from $0.58$ to $0.7$ with a step size $0.02$ (other parameters remain the same).  
The SNR of the simulated network hinges on the smallest eigenvalue of $P$, which in turn hinges on $\beta_2$.  
%
%
%
The results (based on $20$ repetitions) are in Figure~ \ref{fig:compare_nonconvex} (right), 
which suggest that R-SCORE steadily outperforms npMLE for $\beta_2$ in the entire range. 
Also,  R-SCORE is much faster than npMLE. In each repetition, it takes R-SCORE only 6 seconds, whereas it takes the npMLE more than 300 seconds (more than 50 times longer).  This shows that 
R-SCORE not only is significantly faster than npMLE, but may also outperform the npMLE in 
many network settings. 
\begin{figure}[htbp]
\centering
\includegraphics[width=0.8\linewidth, height = 0.28\linewidth]{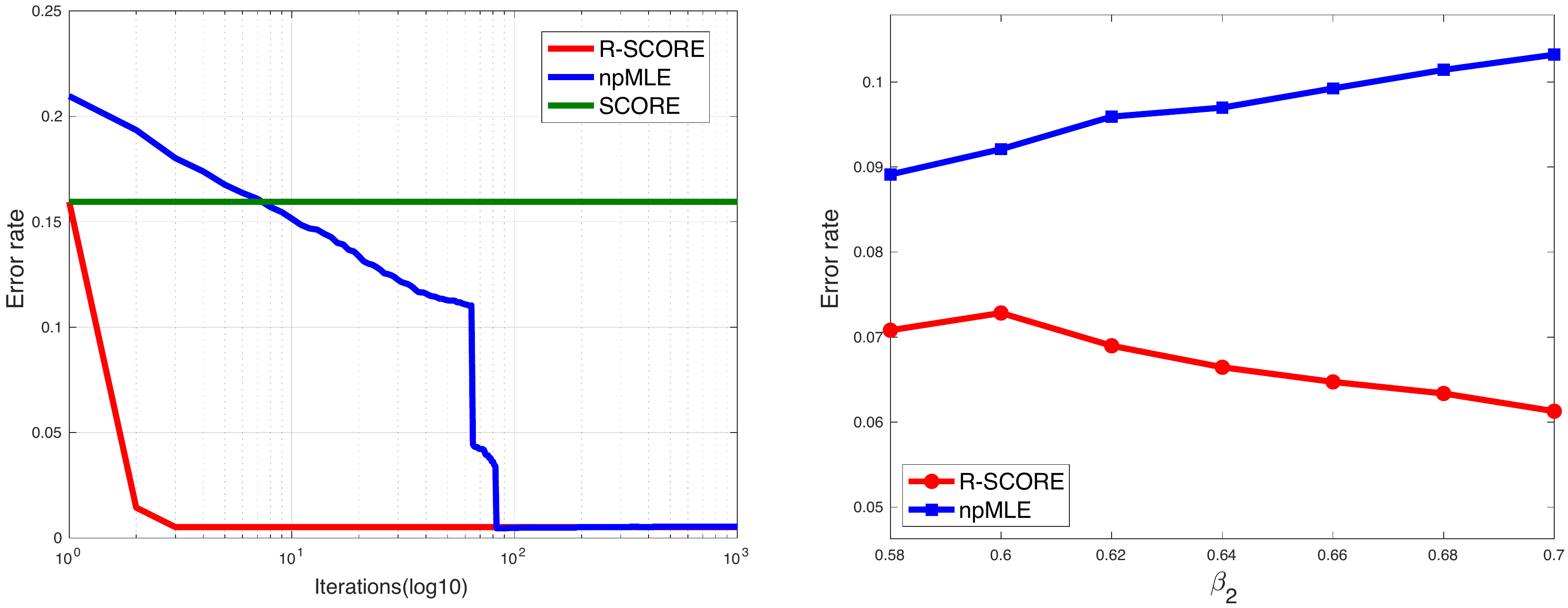}
\caption{Left panel: The error rates  by R-SCORE and npMLE for different $m$ ($\#$ of iterations). See Experiment 2  ($x$-axis: log$_{10}(m)$, $y$-axis: error rates). SCORE is also included for comparison, which corresponds to R-SCORE with $m = 0$.  Right panel: The error rates  by R-SCORE and npMLE for different $\beta_2$. See Experiment 3 for setting details ($x$-axis: $\beta_2$; $y$-axis: error rate).}
\label{fig:compare_nonconvex} 
\end{figure}

\section{Discussion} \label{sec:Discu} 
In this paper, we have made a three-fold contribution to the area of network community detection. 
First, we  propose the logit-DCBM as a new network model. We argue that the logit-DCBM is more 
reasonable than the popular DCBM, but also poses a challenge.  
Second, to overcome the challenge, we propose a trick that can effectively cancel  the effect of the nonlinear factors of the logit-DCBM model in some statistics (especially the ratio of two cycle-counts).  Last, we propose R-SCORE as a new algorithm for community detection, and show that it can significantly improve over existing spectral approaches including   the SCORE.  Our idea is generalizable to many other settings.  For example, the 
$p_1$ model by Holland and Leinhardt \cite{p1model} is one of the most popular models for 
directed networks with $1$  community.  Following the idea here,   we can generalize it to a model with multiple communities, 
and extend R-SCORE for community detection with the new model. 
Also for example, the cancellation trick can be extended to many other settings (e.g.,  analysis of the $p_1$ model for directed networks \cite{p1model}, 
text analysis \cite{Ke-AR},  tensor analysis \cite{Yuan2018TestHyper}) where the data matrix $A$ satisfies $\mathbb{E}[A] = N \circ \widetilde{\Omega}$ for a simple low-rank matrix $\widetilde{\Omega}$ and a matrix $N$ consisting nonlinear factors. 
Given that nonlinear models become increasingly more important in statistics and machine learning, the trick (and its extended form) may find increasingly more uses  in many applications in the near future.  
 
The cancellation trick is especially useful.  In machine learning,  we have many nonlinear latent variable models 
spreading in many areas (e.g., cancer clustering \citep{IFPCA},  
text analysis \citep{Ketopic},    and empirical finance).  Due to the nonlinearity, 
how to analyze such models is a challenging problem. 
In this paper, we propose an interesting cancellation trick using which 
we can effectively remove the nonlinear factor in some latent variable models.   
For space reasons,  we only showcase this trick with a network setting, but the idea 
is extendable to other nonlinear latent space models. For this reason, our work 
may spark new research in many different directions in machine learning.

%

\appendix

\section{Notations}

Throughout this supplementary material, we will adopt the following notions. (1) Let $\{e_k\}_{k=1}^K$ be the standard basis of $\mathbb R^{K}$. To distinguish, we use $\{e_{n,i}\}_{i=1}^n$ to denote the standard basis of $\mathbb R^{n}$. (2) We write ${\bf 1}_m$ the all-one vector of dimension $m$. (3) For  two sequence of numbers $a_n, b_n>0$ depending on $n$, we write $a_n \gg b_n$ or $b_n\ll a_n$ if  $b_n/a_n = o(1)$ as $n\to \infty$; and $a_n \asymp b_n$ is there exists constants $C,c>0$ such that $cb_n<a_n \leq Cb_n $.  (4) Let $O(K-1)$ be group of all  $(K-1)\times (K-1)$ orthogonal matrices. (5) For any matrix $M\in \mathbb R^{m\times m}$, let its SVD be $M = UDV'$. We adopt the notion that ${\rm sgn}(M) = UV'$. (6) We denote the $(i,j)$-th entry of a matrix $M$ as $M(i,j)$ or $M_{ij}$, and  the $i$-th component of a vector $u$ as $u(i)$ or $u_i$.  (7) We denote $c, C$ the generic constants which may vary from line to line.

\section{The error rate of SCORE}
\subsection{Proof of Theorem~\ref {thm:SCORE}}
The proof of Theorem~\ref {thm:SCORE} can be separated into two parts. First, we connect the Hamming error with the error rates of SCORE vectors $R$ where 
\[
R = {\rm diag} (\xi_1)^{-1} (\xi_2, \cdots, \xi_K)
\]
with $\xi_k$ denoted as the eigenvector  associated with the $k$-th largest eigenvalue ( in magnitude) of $\widetilde \Omega$, for $1\leq k \leq K$. The result is collected in the following lemma and the proof is postponed to next subsection. 

\begin{lemma} \label{lem: Hamming-R}
Let $\widehat R$ be the SCORE vectors obtained from the observed network (either $A$ or $A\oslash \widehat N$). Denote by $R$, the counterpart for $\widetilde\Omega$. Suppose that $\min_{\mathcal O \in O(K-1)} \| \widehat R \mathcal O - R\|_F^2 = o(n)$. Then, the Hamming error $r_n$ satisfies 
\[
r_n = n^{-1} \sum_{i=1}^n \|\hat \pi_i - \pi_ i \|_1 \leq n^{-1}  \min_{\{\mathcal O \in O(K-1)\}} \| \widehat R \mathcal O - R\|_F^2
\]
\end{lemma}

Next, we claim the error rate $n^{-1}  \min_{\{\mathcal O \in O(K-1)\}} \| \widehat R \mathcal O - R\|_F^2$ by applying SCORE algorithm. The key technical component is to conduct delicate eigenvector analysis and especially employ leave-one-out technique to derive sharp entry-wise eigenvector bounds for $\xi_1$. We present the result below, and  the proof is relegated into Section~\ref{sub:R-SCORE-rate}.

\begin{lemma}\label{lem: RSCORE}
Let $\widetilde R$ denote the SCORE vectors by employing SCORE directly on $A$. Under the assumptions in Theorem~\ref {thm:SCORE},
it holds that with probability $1- o(n^{-3})$,
\[
n^{-1}\min_{\mathcal O \in O(K-1)} \| \widetilde R \mathcal O - R\|_F^2 \leq C\frac{  \| (N-{\bf 1}_n{\bf 1}_n')  \circ \widetilde \Omega\|^2  +   { \lambda_1(\widetilde \Omega) } }{\big| \lambda_K(\widetilde \Omega) \big|^2}
\]
\end{lemma}

Therefore, Theorem~\ref {thm:SCORE} follows directly from Lemmas~\ref{lem: Hamming-R} and \ref{lem: RSCORE}. In particular, if $A$ satisfies DCBM, by definition, $N- {\bf 1}_n {\bf 1}_n'$. Then, $\| (N-{\bf 1}_n{\bf 1}_n')  \circ \widetilde \Omega\| = 0$, which yields that 
\[
r_n(\widehat{\Pi}^{score}) \leq C \frac{\lambda_1(\widetilde{\Omega})}{|\lambda_K(\widetilde{\Omega})|^2}. 
\]

To complete the proof, we show the proofs of Lemmas~\ref{lem: Hamming-R} and \ref{lem: RSCORE} in the subsequent two subsections. 
\subsection{Proof of Lemma~\ref{lem: Hamming-R}}
The proof is similar to  the proof of Theorem 2.2 SCORE  \citep{SCORE}, we provide the details below for readers' convenience. 
 Without loss of generality, let us assume $\mathcal O = I_K$ for simplicity. According to  \citep{SCORE}, $R$ contains exactly $K$ distinct rows. Let $r_{(1)}, \cdots, r_{(K)}$ be the $K$ distinct rows in $R$.  To claim the bound, we first show that 
 \[
 \| r_{(k)} - r_{(\ell)}\| \geq c_1
 \]
 for some constant $c_1>0$. To see this, we note that $(\xi_1, \xi_2, \cdots, \xi_K)= (\xi_1, \Xi_1) = \Theta \Pi B$ for some $B = (b_1, b_2, \cdots, b_K) \in \mathbb R^{K\times K}$. Then, it follows that 
 $BB' =( \Pi'\Theta^2 \Pi)^{-1}  = \mathcal P{\rm diag} ([\sum_{i\in \mathcal C_1}\theta_i^2]^{-1}, \cdots, [\sum_{i\in \mathcal C_K}\theta_i^2]^{-1}) \mathcal P' $ for some permutation matrix $\mathcal P $. Thanks to the conditions that $\theta_i \asymp \bar \theta$ and $n_k \asymp n$ for all $1\leq k \leq K$,   the conditional number of $BB'$ is constant and 
 \[
 \lambda_{\min} (BB')\asymp \lambda_{\max} (BB')\asymp \frac{1}{n\bar \theta^2}
 \]
 In particular, as $\xi_1 = \Theta \Pi b_1 $, it is not hard to derive from $\Theta \Pi P \Pi' \Theta \xi_1 = \lambda_1 \Theta \Pi b_1$ that 
 \[
 P \Pi' \Theta^2 \Pi b_1 = \lambda_1b_1
 \]
 Therefore, $b_1 $ is the first right eigenvector of $P (\Pi' \Theta^2 \Pi)  $.
 Under  the condition in  (\ref{asm:eta}), $|b_1(k) | \asymp 1/\sqrt{n\bar \theta^2}$ for all $1\leq k \leq K$. As a result, 
 \begin{align}\label{bdd:eigenvector}
 |\xi_1(i)|\asymp \theta_i /\sqrt{n\bar \theta^2} \asymp 1/\sqrt n, \qquad \|e_{n,i}'\Xi_1 \|\leq C/\sqrt n
 \end{align}
  and 
 \[
\lambda_{\min}( {\rm diag} (b_1)^{-1}  B ) \geq c_0 
 \]
 for some $c_0>0$. 
Notice that 
\[
 R = \Pi [{\rm diag} (b_1)^{-1}  (b_2, \cdots, b_K ) ]= \Pi (r_{(1)}', \cdots, r_{(K)}')'
 \]
Therefore, 
\[
\| r_{(i)} - r_{(j)} \| = \| e_i'  {\rm diag} (b_1)^{-1}  B-e_j' {\rm diag} (b_1)^{-1}  B \| \geq \sqrt 2\lambda_{\min}( {\rm diag} (b_1)^{-1}  B ) \geq  \sqrt 2c_0
\]
for some $c_0>0$. 
We define $c_1 = 2c_0$. Let $V_1, \cdots, V_K$ denote the disjoint index sets corresponding to $r_{(1)}, \cdots, r_{(K)}$. The K-means algorithm aims to find a partition of the nodes $S^* = (S_1, S_2, \cdots, S_K)$ such that 
\[
S^* = {\rm argmin} \sum_{k=1}^K \sum_{i \in S_k } \|\hat r_i - m_k\|^2, \quad m_k = \sum_{i\in S_k}\hat r_i  \text{ for $1\leq k \leq K$}.
\]
Define the output centers are $m_1^*, \cdots, m_K^*$. We introduce a matrix $M = (m_1', \cdots, m_n')' $ such that 
\[
m_i= m_k^* ,\qquad \text{ if \quad $i \in S_k$ }
\]
Thus, 
\[
\|\widehat R - M \|_F^2 \leq \|\widehat R - R \|_F^2 \quad \text{and } \quad \|R - M \|_F^2 \leq 4 \|\widehat R - R \|_F^2
\]
Let $\mathcal I : = \{ i: \|\hat r_i - r_i \|\leq \sqrt 2 c_0/8,  \|m_i - r_i\|\leq  \sqrt 2 c_0/8 \}$ and $\mathcal I_k = \mathcal I \cap V_k$ for $1\leq k \leq K$. We first prove that the nodes in $\mathcal I$ are correctly recovered. It suffices to show that for any $i\in \mathcal I_k, j\in \mathcal I_\ell$, 
\begin{align}\label{2024092201}
m_i = m_j  \quad \text{if and only if} \quad k = \ell.
\end{align}
To see this, consider $k\neq \ell$, then $\|r_k - r_\ell \|\geq \sqrt 2c_0$. It further yields that for $i\in \mathcal I_k, j\in \mathcal I_\ell$
\[
\|m_i - m_j\|  \geq \|r_i - r_j\| - \|m_i - r_i \| - \|m_j- r_j\| \geq \sqrt 2c_0 \cdot 3/4
\]
Suppose that  $\mathcal I_k \neq \emptyset$ for all $1\leq k\leq K$,  then for every $k$, we select a point $i_k$ and its corresponding $m_{i_k}$. It follows that 
 \[
\| m_{i_k} - m_{i_\ell} \|\geq \|r_{i_k} - r_{i_{\ell}}\| - \|m_{i_k} - r_{i_k} \| - \|m_{i_\ell}- r_{i_\ell}\| \geq \sqrt 2c_0 \cdot 3/4
 \] 
By doing so, we fix the $K$ distinct rows in $M$. Thus, based on the above arguments,  for any two rows in $M$, their $\ell_2$ norm distance is either $0$ or larger than $\sqrt 2c_0 \cdot 3/4$. For any $i,j \in \mathcal I_k$, since
 \[
 \|m_i - m_j \|\leq \|m_i - r_k \| + \|m_j - r_k\| \leq \sqrt 2c_0 /4,
 \]
 it must hold that $m_i = m_j$. 
 
 To complete the proof of (\ref{2024092201}), we need to  claim $\mathcal I_k \neq \emptyset$ for all $1\leq k\leq K$.  We will prove by contradiction.  Suppose there exist $k_0$ such that $\mathcal I_{k_0} = \emptyset$. It follows that 
 \[
 \sum_{i \in  V_{k_0}}\|\hat r_ i - r_i\|^2 + \|m_i - r_i\|^2\geq |V_{k_0}| c_0/32 \geq \tilde c n
 \]
 under the assumption that $n_k \asymp n$ for all $1\leq k\leq K$. This implies that $\|\widehat R  - R \|_F^2 + \|M - R\|_F^2\geq \tilde c n$. Moreover, 
 \[
 \|\widehat R  - R \|_F^2\geq \tilde c n/5
 \]
 which contradicts to $\|\widehat R  - R \|_F^2\ll n$. As a result, $\mathcal I_k \neq \emptyset$ for all $1\leq k\leq K$. We thus finish the proof that nodes in $\mathcal I $ are exactly recovered. 
 
 Next, to finish the proof,  we show that 
 \[
 |\mathcal I^c| \leq \|\widehat R  - R \|_F^2
 \]
 Note that for $i \in \mathcal I^c$, either $\|\hat r_i - r_i \|> \sqrt 2 c_0/8$ or $ \|m_i - r_i\|>  \sqrt 2 c_0/8$. Since $\|M - R\|_F^2\leq 4\|\widehat R  - R \|_F^2$, we can obtain that 
 \[
  |\mathcal I^c| \leq  \frac{\|\widehat R  - R \|_F^2}{( \sqrt 2 c_0/8)^2 } +  \frac{\|M - R\|_F^2}{( \sqrt 2 c_0/8)^2} \leq  \frac{160}{c_0^2}\|\widehat R  - R \|_F^2
 \]
 Consequently, 
 \[
 r_n = \sum_{i\in \mathcal I} \|\hat \pi_ i - \pi_ i \|_1 +  \sum_{i\in \mathcal I^c} \|\hat \pi_ i - \pi_ i \|_1  \leq 2|\mathcal I^c| \leq C_1 \|\widehat R  - R \|_F^2
 \]
 We thereby conclude the proof.

\subsection{Proof of Lemma~\ref{lem: RSCORE}} \label{sub:R-SCORE-rate}
We define $(\tilde  \lambda_k,\tilde \xi_k)$ be the $k$-th largest eigen-pair of $A$ (in magnitude) for $1\leq k \leq K$. For simplicity, write $\widetilde \Xi_1 := (\tilde \xi_2, \cdots, \tilde \xi_K)$. Without loss of generality, we assume that ${\rm sgn}(\tilde \xi_1'\xi_1)= 1$.  Let $\mathcal O: = {\rm sgn} (\widetilde \Xi_1' \Xi_1)$. By definition, 
\begin{align}\label{R:2024092101}
\|\widetilde R\mathcal O - R \|_F^2 &=\|  {\rm diag}(\tilde \xi_1)^{-1} (\widetilde \Xi_1\mathcal O - \Xi_1) -  \big[ {\rm diag}(\tilde \xi_1)^{-1}  - {\rm diag}( \xi_1)^{-1} \big] \Xi_1\|^2_F \notag\\
& \leq C\Big(\|  {\rm diag}(\tilde \xi_1)^{-1} (\widehat \Xi_1 \mathcal O - \Xi_1) \|_F^2 + \|  \big[ {\rm diag}(\tilde \xi_1)^{-1}  - {\rm diag}( \xi_1)^{-1} \big] \Xi_1\|_F^2\big) \notag\\
& \leq C\Big( \|  {\rm diag}(\tilde \xi_1)^{-1} (\widetilde \Xi_1 \mathcal O - \Xi_1) \|_F^2 + \|\tilde \xi_1 - \xi_1\|^2\|{\rm diag}(\tilde \xi_1)^{-1}   {\rm diag}( \xi_1)^{-1}\Xi_1\|_{2\to \infty}^2 \Big) 
\end{align}
According to the RHS, we need to prove an upper bounds for $\|\widetilde \Xi_1 \mathcal O - \Xi_1 \|_F$ and $\|\tilde \xi_1 - \xi_1\|$, and further show that $|\tilde \xi_1(i) |\asymp 1/\sqrt n$ for $1\leq i \leq n$.

First, we claim upper bounds for $\|\widetilde \Xi_1 \mathcal O - \Xi \|_F$ and $\|\tilde \xi_1 - \xi_1\|$. Using sine-theta theorem \citep{sin-theta, yu2015useful}, we have 
\[
\min\{\| \tilde \xi_1 - \xi_1\|, \|\Xi_1'\tilde \xi_1\|\}\leq C\frac{\|A - \widetilde \Omega\|}{\lambda_1(\widetilde \Omega) } 
\qquad
 \| \widetilde \Xi_1 O- \Xi_1 \|_F \leq C\frac{\|A - \widetilde \Omega\|}{ |\lambda_K(\widetilde \Omega) |}
\]
Since $A = \widetilde \Omega  +  (N -{\bf 1}_n {\bf 1}_n') \circ \widetilde \Omega - {\rm diag} (N\circ \widetilde \Omega) + W =  \widetilde \Omega + \widetilde W$, we thus bound 
\[
\|A - \widetilde \Omega\| \leq \| (N -{\bf 1}_n {\bf 1}_n') \circ \widetilde \Omega \| + \|{\rm diag} (N\circ \widetilde \Omega) \| + \| W\| \leq  \| (N -{\bf 1}_n {\bf 1}_n') \circ \widetilde \Omega \|  + C\sqrt{n\bar \theta^2}
\]
with probability $1- o(n^{-3})$.
Here we used the derivation
\[
\|{\rm diag} (N\circ \widetilde \Omega) \| \leq \|{\rm diag} (\widetilde \Omega) \| \leq C\bar \theta^2, \qquad \|W\|\leq C\sqrt{n\bar \theta^2}
\]
by the non-asymptotic bounds on the norm of random matrices in \citep{bandeira2016sharp}. It is worth mentioning that $\lambda_1(\widetilde \Omega) = \lambda_1(P (\Pi'\Theta^2 \Pi))\asymp n\bar \theta^2$. As a result, 
\begin{align}\label{20240922001}
&\min\{\| \tilde \xi_1 - \xi_1\|, \|\Xi_1'\tilde \xi_1\|\}\leq C\frac{ \| (N -{\bf 1}_n {\bf 1}_n') \circ \widetilde \Omega \| + \sqrt{\lambda_1(\widetilde \Omega)}}{\lambda_1(\widetilde \Omega) } 
\notag\\
& \| \widetilde \Xi_1 O- \Xi_1 \|_F \leq C\frac{ \| (N -{\bf 1}_n {\bf 1}_n') \circ \widetilde \Omega \| + \sqrt{\lambda_1(\widetilde \Omega)}}{ |\lambda_K(\widetilde \Omega) |}
\end{align}

Next, we aim to show that $|\tilde \xi_1(i) |\asymp 1/\sqrt n$ for $1\leq i \leq n$. Given that $| \xi_1(i) |\asymp 1/\sqrt n$ for $1\leq i \leq n$, it suffices to show that $\|\tilde \xi_1 - \xi_1\|_{\max} \ll 1/\sqrt n$. To see this, we consider the eigen-perturbation that 
\begin{align*}
 \tilde \xi_1 - \xi_1=( \tilde \lambda_1^{-1} \lambda_1 \xi_1' \tilde \xi_1 - 1)\xi_1 + \tilde \lambda_1^{-1} \Xi_1 {\rm diag} (\lambda_2, \cdots, \lambda_K)  \Xi_1' \tilde \xi_1 +  \tilde \lambda_1^{-1}\widetilde W \tilde \xi_1\,.
\end{align*}
 By the first inequality in (\ref{20240922001}) and the Weyl's inequality, we bound 
 \begin{align*}
 | \tilde \lambda_1^{-1} \lambda_1 \xi_1' \tilde \xi_1 - 1| &\leq C\Big(   \frac{|\tilde \lambda_1 - \lambda_1|}{\lambda_1} + |  \xi_1' \tilde \xi_1 - 1| \Big) \leq C\Big( \frac{\|A - \widetilde\Omega \|}{\lambda_1} + \|\tilde \xi_1- \xi_1\|^2\Big) \notag\\
 & \leq C\frac{ \| (N -{\bf 1}_n {\bf 1}_n') \circ \widetilde \Omega \| + \sqrt{\lambda_1(\widetilde \Omega)}}{\lambda_1(\widetilde \Omega) } 
 \end{align*}
 and 
  \begin{align*}
 \| \tilde \lambda_1^{-1} {\rm diag} (\lambda_2, \cdots, \lambda_K)  \Xi_1' \tilde \xi_1\| &\leq \|\Xi_1' \tilde \xi_1\| \leq C\frac{ \| (N -{\bf 1}_n {\bf 1}_n') \circ \widetilde \Omega \| + \sqrt{\lambda_1(\widetilde \Omega)}}{\lambda_1(\widetilde \Omega) } 
 \end{align*}
 Based on these, we arrive at 
 \[
| \tilde \xi_1(i)  - \xi_1(i) | \leq \frac{C}{\sqrt n} \frac{ \| (N -{\bf 1}_n {\bf 1}_n') \circ \widetilde \Omega \| + \sqrt{\lambda_1(\widetilde \Omega)}}{\lambda_1(\widetilde \Omega) }  + \frac{ \big| e_{n,i} ' \widetilde W \tilde \xi_1 \big| }{n\bar \theta^2}
 \]
 for $1\leq i \leq n$,
 due to the fact that $|\tilde \lambda_1- \lambda_1|\leq \|A - \widetilde \Omega\|\ll \lambda_1\asymp n\bar \theta^2$ and $\max{\|\xi_1\|_{\max}, \max_i \|e_{n,i}'\Xi\|}\leq C/\sqrt n $. Then, it suffices to derive an upper bound for  $e_{n,i} ' \widetilde W \tilde \xi_1$. We first decompose 
 \[
 |e_{n,i}  ' \widetilde W \tilde \xi_1| \leq  |e_{n,i} ' (N -{\bf 1}_n {\bf 1}_n') \circ \widetilde \Omega \tilde \xi_1| +|e_{n,i}'  {\rm diag} (N\circ \widetilde \Omega)\tilde \xi_1|  + |e_{n,i}' W \tilde \xi_1| \leq \bar \theta^2 | \tilde \xi_1(i) | +  |e_{n,i}' W \tilde \xi_1|
 \]
 Let $ \tilde \xi_1^{(i)}$ be the first eigenvector of $A^{(i)} = \Omega - {\rm diag} (\Omega) + W^{(i)}$ where $W^{(i)}$ is obtained by zeroing out the $i$-th row and column of $W$. Then, 
 \begin{align}\label{202409220004}
 |e_{n,i}'  W \tilde \xi_1| \leq |e_{n,i}'  W \tilde \xi^{(i)}_1 | + \sqrt{n\bar \theta^2} \| \tilde \xi_1 - \tilde \xi_1^{(i)}\|
 \end{align}
 By Bernstein inequality, we bound 
 \begin{align}\label{202409220006}
  |e_{n,i}'  W \tilde \xi^{(i)}_1 | &\leq C(\sqrt{\bar\theta^2\| \tilde \xi^{(i)}_1 \|^2 \log(n) } + \|\tilde \xi_1^{(i)}\|_{\max} \log(n) )   \notag\\
  & \leq C(  \bar \theta \sqrt{\log(n)} + \|\tilde \xi_1\|_{\max} \log(n) + \|\tilde \xi_1^{(i)} - \tilde \xi_1\| \log(n)  ) 
 \end{align}
 simultaneously for all $1\leq i \leq n$, with probability $1- o(n^{-3})$.
 
To proceed, we analyze $\|\tilde \xi_1^{(i)} - \tilde \xi_1\|$ below. By sine-theta theorem,
\begin{align}\label{202409220005}
\|\tilde \xi_1^{(i)} - \tilde \xi_1\|& \leq C \frac{\|(A^{(i)} - A)\tilde \xi_1\|}{n\bar \theta^2}\leq C \frac{\|e_{n,i} e_{n,i}'W\tilde \xi_1\|}{n\bar \theta^2} +C \frac{\|We_{n,i} e_{n,i}'\tilde \xi_1\|}{n\bar \theta^2}\notag\\
 & \leq C\frac{| e_{n,i} 'W\tilde \xi_1|}{n\bar \theta^2} + C\frac{\sqrt{n\bar \theta^2} |\tilde \xi_1(i)|}{n\bar \theta^2}
\end{align}
Combining (\ref{202409220004})-(\ref{202409220005}) gives 
\[
 |e_{n,i} '  W \tilde \xi_1| \leq  C(  \bar \theta \sqrt{\log(n)} + \|\tilde \xi_1\|_{\max} \log(n) ) 
\]
Consequently, 
 \begin{align*}
| \tilde \xi_1(i)  - \xi_1(i) | &\leq \frac{C}{\sqrt n} \frac{ \| (N -{\bf 1}_n {\bf 1}_n') \circ \widetilde \Omega \| + \sqrt{\lambda_1(\widetilde \Omega)}}{\lambda_1(\widetilde \Omega) }  + 
\frac{ C| \tilde \xi_1(i) |}{n} \notag\\
& \qquad + \frac{C}{n\bar \theta^2}(  \bar \theta \sqrt{\log(n)} + \|\tilde \xi_1\|_{\max} \log(n)  ) 
 \end{align*}
By decomposing $| \tilde \xi_1(i) | \leq | \xi_1(i) | + | \tilde \xi_1(i) - \xi_1(i)  |$ and $\|\tilde \xi_1\|_{\max}\leq \| \xi_1\|_{\max} + \|\tilde \xi_1 - \xi_1\|_{\max}$, we further have 
\[
| \tilde \xi_1(i)  - \xi_1(i) | \leq \frac{C}{\sqrt n} \frac{\| (N -{\bf 1}_n {\bf 1}_n') \circ \widetilde \Omega \|}{n\bar \theta^2 }   + \frac{C \sqrt{\log(n)}}{n\bar \theta} + \|\tilde \xi_1 - \xi_1\|_{\max} \frac{\log(n)}{n\bar \theta^2} 
\]
Taking maximum and rearranging both sides, it follows that with probability $1- o(n^{-3}) $, 
\[
\|\tilde \xi_1 - \xi_1\|_{\max}\leq \frac{C}{\sqrt n} \frac{\| (N -{\bf 1}_n {\bf 1}_n') \circ \widetilde \Omega \|}{n\bar \theta^2 }   + \frac{C \sqrt{\log(n)}}{n\bar \theta} \ll 1/\sqrt n
\]
under the condition that $\sqrt {n\bar \theta^2} \geq C\log(n)$ and $\|(N - {\bf 1}_n{\bf 1}_n') \circ \widetilde \Omega\|\ll n\bar \theta^2 $. This completes the proof of $|\tilde \xi_1(i) |\asymp 1/\sqrt n$ for $1\leq i \leq n$. 

Therefore, we deduce from (\ref{R:2024092101}), (\ref{20240922001}) that with probability $1-o(n^{-3})$, 
\begin{align*}
\|\widetilde R\mathcal O - R \|_F^2 \leq Cn \Big( \|  (\widetilde \Xi_1 \mathcal O - \Xi) \|_F^2 + \|\tilde \xi_1 - \xi_1\|^2\Big)\leq Cn \frac{ \| (N -{\bf 1}_n {\bf 1}_n') \circ \widetilde \Omega \|^2 + {\lambda_1(\widetilde \Omega)}}{ |\lambda_K(\widetilde \Omega) |^2}
\end{align*}
We thus finish the proof.

\subsection{A remark on $\|(N -{\bf 1}_n {\bf 1}_n') \circ \widetilde \Omega \|$}
We discuss the relation of $\|(N -{\bf 1}_n {\bf 1}_n') \circ \widetilde \Omega \|$ with the eigenvalues of $\widetilde \Omega$ and $\Omega$ in the following lemma. 
\begin{lemma}\label{lem: Relation}
The following inequalities hold.
\begin{align*}
& \| (N-{\bf 1}_n{\bf 1}_n')  \circ \widetilde \Omega\| \leq \|N-{\bf 1}_n{\bf 1}_n'\|_{\max}\lambda_1(\widetilde \Omega) \notag\\
 & |\lambda_{K+1}({\Omega})| \leq  \| (N-{\bf 1}_n{\bf 1}_n')  \circ \widetilde \Omega\|
\end{align*}
\end{lemma}
\begin{proof}
Notice that $c<\min_{i,j}N(i,j) \leq \|N\|_{\max} <1$ for some constant $c>0$. Then, $({\bf 1}_n{\bf 1}_n' - N ) \circ \widetilde \Omega$ is a  symmetric matrix with positive entries. By Perron's theorem (see \citep{HornJohnson} for example), the first eigenvector, denoted by $u_1$,  is a positive vector and $\lambda_1(({\bf 1}_n{\bf 1}_n' - N)  \circ \widetilde \Omega) = \| (N-{\bf 1}_n{\bf 1}_n')  \circ \widetilde \Omega \|$. It follows that 
\begin{align*}
 \| (N-{\bf 1}_n{\bf 1}_n')  \circ \widetilde \Omega \| &= u_1' ({\bf 1}_n{\bf 1}_n' - N)  \circ \widetilde \Omega u_1 = \sum_{i,j}u_1(i) u_1(j) (1- N_{ij}) \widetilde \Omega_{ij}  \notag\\
 &\leq \|N-{\bf 1}_n{\bf 1}_n'\|_{\max} \cdot \sum_{i,j}u_1(i) u_1(j)  \widetilde \Omega_{ij} 
 \notag\\
 & \leq \|N-{\bf 1}_n{\bf 1}_n'\|_{\max}   \lambda_1(\widetilde \Omega) \,. 
\end{align*}
Next, we show the second inequality. Recall  the decomposition 
\[
\Omega = N\circ \widetilde \Omega = \widetilde \Omega + (N - {\bf 1}_n{\bf1}_n')\circ \widetilde \Omega\,. 
\]
By Weyl's inequality (see \citep{HornJohnson} for example), 
\[
|\lambda_{K+1} (\Omega) - \lambda_{K+1}(\widetilde \Omega) |\leq \| \Omega  - \widetilde \Omega  \| \leq \|(N - {\bf 1}_n{\bf1}_n')\circ \widetilde \Omega\|
\]
Since $ \lambda_{K+1}(\widetilde \Omega) = 0$, we conclude that 
\[
 |\lambda_{K+1}({\Omega})| \leq  \| (N-{\bf 1}_n{\bf 1}_n')  \circ \widetilde \Omega\|
\]

\end{proof}

\section{The error rate of R-SCORE}
In this section, we mainly  prove Lemma~\ref{lemma:RSCORE} and Theorem~\ref{thm:main}. We streamline the proofs as follows:
\begin{itemize}
 \item [(1)] We show the error rate of SCORE vectors by R-SCORE, i.e., $\| \widehat R - R\|_F^2$ up to some orthogonal transformation. This, together with Lemma~\ref{lem: Hamming-R} concludes the proof of Lemma~\ref{lemma:RSCORE} (see Section~\ref{subsec:thmSCORE2});
 \item  [(2)] We prove the error rate of refitting $\theta$ and $P$ (see Sections~\ref{subsec:refit_theta} and \ref{subsec:refit_P}); 
 \item [(3)] Third, we investigate the error rate of $N$, more precisely,  $\|( N\oslash \widehat N - {\bf 1}_n{\bf 1}_n' ) \circ \widetilde \Omega \|$ and $\|N\oslash \widehat N - {\bf 1}_n{\bf 1}_n'  \|_F$ (see Section~\ref{subsec:refit_N});
 \item [(4)] Combining all the previous results, together with  Lemma~\ref{lem: Hamming-R}, we complete the proof of Theorem~\ref{thm:main} (see Section~\ref{subsec:main_thm} );
  \item [(5)] We also provide the brief proof of the Corollary~\ref{cor:main}, as it follows simply from Theorem~\ref{thm:main}. 
 \end{itemize}
 The details are provided in the subsequent subsections. 
 
 \subsection{Proof of Lemma~\ref{lemma:RSCORE}} \label{subsec:thmSCORE2}

Recall the assumption that $\widehat N$ satisties 
\[
{\bf1}_n {\bf 1}_n' \oslash \widehat N - {\bf1}_n {\bf 1}_n' = \widehat \Theta \widehat \Pi \widehat P \widehat \Pi' \widehat \Theta  
\]
such that with probability $1- o(n^{-3})$, 
\[
\|\widehat P - P\|_{\max}\ll \min\{ 1, \lambda_{\min}(P) \bar \theta^{-1}\}, \quad  \|\widehat \Pi - \Pi\| (\sqrt{n}\, \lambda_{\min}(P))^{-1} \bar\theta \to 0, 
\]
and 
\[
 \hat \theta_i \leq C\bar \theta, \quad  \bar \theta = o(1) 
\]
for some constant $c, C>0$. It follows that $\widehat N_{ij} = (1+ \hat \theta_i \hat \theta_j \hat \pi_i' \widehat P \hat \pi_j )^{-1} >C$ for some constant $0<C<1$  and $\widehat N_{ij} \leq 1$ for all $0<i,j<n $.

Let  $(\hat \lambda_k, \hat \xi_k)$ be the $k$-th largest eigen-pair of $A\oslash \widehat N$ (in magnitude) for $1\leq k \leq K$. For brevity, write $\widehat \Xi_1: = (\hat \xi_2, \cdots, \hat \xi_K)$. Denote by $(\lambda_k, \xi_k)$ and $\Xi_1$ the counterparts for the low-rank matrix $\widetilde \Omega = \Theta \Pi P \Pi'\Theta$. Without loss of generality, we assume both $\hat \xi_1$ and $\xi_1$ are positive. Under these notations, the SCORE vectors of $A$ and $\widetilde \Omega$ are defined as
\begin{align*}
\widehat R = (\hat r_1, \hat r_2, \cdots, \hat r_n)' = {\rm diag}(\hat \xi_1)^{-1} \widehat \Xi_1, \qquad R = (r_1, r_2, \cdots, r_n)' = {\rm diag}( \xi_1)^{-1}  \Xi_1
\end{align*}
We bound the error of $\widehat R - R$ by the eigenvalues of $A$. Consider the SVD of $\widehat \Xi_1' \Xi = UDV'$. Define $ O :={\rm sgn}( \widehat \Xi_1' \Xi_1 )= UV'$. Our model assumptions  gives that $\lambda_1(\widetilde \Omega) - |\lambda_2(\widetilde \Omega) | \geq c \lambda_1(\widetilde \Omega)$.  Applying sine-theta theorem \citep{sin-theta, yu2015useful}, we have 
\[
\min\{\| \hat \xi_1 - \xi_1\|, \|\Xi_1'\hat \xi_1\|\}\leq C\frac{\|A\oslash \widehat N - \widetilde \Omega\|}{\lambda_1(\widetilde \Omega) } 
\qquad
 \| \widehat \Xi_1 O- \Xi_1 \|_F \leq C\frac{\|A\oslash \widehat N - \widetilde \Omega\|}{ |\lambda_K(\widetilde \Omega) |}
\]

We write  
\[
A\oslash \widehat N =\widetilde \Omega +  (N\oslash \widehat N -{\bf 1}_n {\bf 1}_n') \circ \widetilde \Omega - {\rm diag} (N\circ \widetilde \Omega \oslash \widehat N) + W\oslash \widehat N : = \widetilde \Omega + \widetilde W
\]
It follows that 
\begin{align*}
\|\widetilde W\| & \leq \| (N\oslash \widehat N -{\bf 1}_n{\bf 1}_n') \circ \widetilde \Omega \| + \|{\rm diag} (N\circ \widetilde \Omega \oslash \widehat N) \| + \| W\oslash \widehat N\| \notag\\
& \leq  \| (N\oslash \widehat N -{\bf 1}_n{\bf 1}_n') \circ \widetilde \Omega \|   +  \| W\oslash N \circ(N\oslash \widehat N - {\bf 1}_n {\bf 1}_n')\|+C\sqrt{n\bar \theta^2}
\end{align*}
To obtain the RHS, we bound
\begin{align*}
\|{\rm diag} (N\circ \widetilde \Omega \oslash \widehat N) \|_F \leq C\|\widetilde \Omega\|_{\max} \leq C  \bar \theta^2 
\end{align*}
and 
\begin{align*}
\| W\oslash \widehat N\| &\leq  \| W\oslash  N\| +  \| W\oslash N \circ(N\oslash \widehat N - {\bf 1}_n {\bf 1}_n')\|   \notag\\
&\leq C\sqrt{n \bar \theta^2} + \| W\oslash N \circ(N\oslash \widehat N - {\bf 1}_n {\bf 1}_n')\|  
\end{align*}
where we used non-asymptotic bounds on the norm of random matrices in \citep{bandeira2016sharp} since $W\oslash N$ is a symmetric random matrix with independent upper triangular entries and each entry in $N$ of constant order. 
We further study $ \| W\oslash N \circ(N\oslash \widehat N - {\bf 1}_n {\bf 1}_n')\|$ as follows. Notice that 
\[
N\oslash \widehat N - {\bf 1}_n{\bf 1}_n' = ( \widehat \Theta \widehat \Pi \widehat P \widehat \Pi' \widehat \Theta  - \Theta \Pi P \Pi'\Theta) \circ N 
\]
by the definition of $N$ and $\widehat N$. Therefore, it suffices to bound 
\[
\|W\circ ( \widehat \Theta \widehat \Pi \widehat P \widehat \Pi' \widehat \Theta  - \Theta \Pi P \Pi'\Theta)\|
\]
Next, we decompose 
\begin{align*}
& \quad \|W\circ (\widehat \Theta \widehat \Pi \widehat P \widehat \Pi' \widehat \Theta  - \Theta \Pi P \Pi'\Theta)\|  \notag\\
&  =  \| W\circ \widehat \Theta \widehat \Pi (\widehat P - P)  \widehat \Pi' \widehat \Theta \|
+  \|W\circ  \widehat \Theta (\widehat \Pi - \Pi) P \widehat \Pi'\widehat \Theta \| + \| W\circ\widehat \Theta \Pi P (\widehat \Pi - \Pi)'\widehat \Theta \|  \notag\\
& \qquad   + \| W\circ  ( \widehat \Theta - \Theta)  \Pi P  \Pi'\widehat \Theta\|  + \|  W\circ \Theta \Pi P \Pi'( \widehat \Theta - \Theta) \|\notag\\
& =: \mathcal T_1 +  \mathcal  T_2 + \mathcal T_3 +  \mathcal  T_4+ \mathcal T_5 
\end{align*}
We bound each term separately below. 
For $\mathcal T_1$, we have 
\begin{align*}
\mathcal T_1 = \| \widehat \Theta (  W\circ  \widehat \Pi (\widehat P - P)  \widehat \Pi' ) \widehat \Theta   \| \leq C\bar \theta^2 \| W\circ  \widehat \Pi (\widehat P - P)  \widehat \Pi'  \| &\leq C\bar \theta^2 \| \widehat P - P\|_{\max} \|W\|_F \notag\\
&\leq n\bar \theta^3 \| \widehat P - P\|_{\max} 
\end{align*}
where $\|W\|_F \leq \sqrt n \|W\|\leq C n\bar \theta$ with probability $1- o(n^{-3})$. 

The analysis for bounding  $\mathcal T_2$ and  $\mathcal T_3$ is similar,  we provide the details for $\mathcal T_2$ only. 
\begin{align*}
\mathcal T_2 = \|  \widehat \Theta( W\circ  (\widehat \Pi - \Pi) P \widehat \Pi' ) \widehat \Theta \|  &\leq C\bar \theta^2 \| W\circ  (\widehat \Pi - \Pi) P \widehat \Pi'  \| \leq C\bar \theta^2 \| W\circ  (\widehat \Pi - \Pi) P \widehat \Pi'  \|_F \notag\\
& \leq C\bar \theta^2\sqrt{\sum_{i: \hat \pi_i \neq \pi_i} \sum_j W_{ij}^2[ (\hat \pi_i - \pi_i)'P \hat \pi_j ]^2}\notag\\
& \leq C\bar \theta^2\sqrt{\sum_{i: \hat \pi_i \neq \pi_i} \sum_j W_{ij}^2 \cdot 4\|P\|_{\max}^2}\notag\\
 & \leq C\bar \theta^2\|W\|_{2\to \infty} \| \widehat \Pi - \Pi\|  \notag\\
& \leq  \bar \theta^2\sqrt{n\bar \theta^2} \| \widehat \Pi - \Pi\|
\end{align*}
The last step is due to fact that $\|e_{n,i}' W\| \leq \sqrt{n\bar \theta^2}$ simultaneously for all $1\leq i \leq n$ with probability $1- o(n^{-3})$ by Bernstein inequality. From here to the end of this subsection, with a slight abuse of notation, we will use $\{e_i\}_{i=1}^n$ to denote the standard basis of $\mathbb R^n$ for simplicity. 

Next, for $\mathcal T_4$ and $\mathcal T_5$, the analysis is also analogous, and we show the bound for $\mathcal T_4$ in detail and omit the proof for $\mathcal T_5$. 
\begin{align*}
\mathcal T_4 =\| ( \widehat \Theta - \Theta) ( W\circ    \Pi P  \Pi') \widehat \Theta\|   \leq C\bar \theta^2 \|   W\circ    \Pi P  \Pi' \| \leq C\bar \theta^2 \sqrt{n\bar\theta^2} 
\end{align*}
where we bound $ \|   W\circ    \Pi P  \Pi' \|\leq \sqrt{n\bar \theta^2} $ by non-asymptotic bound for random matrices since $\|\Pi P  \Pi' \|_{\max} \leq C$ and $W\circ    \Pi P  \Pi' $ is symmetric random matrix with independent upper triangular entries.  

Combining the discussions above, we have 
\[
\|W\circ (\widehat \Theta \widehat \Pi \widehat P \widehat \Pi' \widehat \Theta  - \Theta \Pi P \Pi'\Theta)\|  \leq C\Big(  n\bar \theta^3 \| \widehat P - P\|_{\max}  +  \bar \theta^2\sqrt{n\bar \theta^2} \| \widehat \Pi - \Pi\| + \bar \theta^2 \sqrt{n\bar\theta^2}\Big)
\]
This further gives rise to 
\begin{align*}
\|\widetilde W\| & \leq  \| (N\oslash \widehat N -{\bf 1}_n{\bf 1}_n') \circ \widetilde \Omega \|   +  \|W\circ (\widehat \Theta \widehat \Pi \widehat P \widehat \Pi' \widehat \Theta  - \Theta \Pi P \Pi'\Theta)\| +C\sqrt{n\bar \theta^2} \notag\\
& \leq  \| (N\oslash \widehat N -{\bf 1}_n{\bf 1}_n') \circ \widetilde \Omega \|   +C\Big(  n\bar \theta^3 \| \widehat P - P\|_{\max}  +  \bar \theta^2\sqrt{n\bar \theta^2} \| \widehat \Pi - \Pi\| +  \sqrt{n\bar\theta^2}\Big) \notag\\
&\ll | \lambda_{K} (\widetilde \Omega) |
\end{align*}
under the assumptions that 
\[
\|\widehat P - P \|_{\max} \ll \lambda_{\min} (P)/\bar \theta \quad \| \widehat \Pi - \Pi\| \ll \sqrt{n} \lambda_{\min}(P)/\bar \theta\quad   \| (N\oslash \widehat N -{\bf 1}_n{\bf 1}_n') \circ \widetilde \Omega \| \ll | \lambda_{K} (\widetilde \Omega) |
\]
and $\sqrt{n\bar \theta^2} \lambda_{\min} (P)  \geq c_3 \log(n) $\,. 

We note that $\lambda_1(\widetilde \Omega)\asymp {n\bar\theta^2} $ and $|\lambda_K(\widetilde \Omega) | \asymp n\bar \theta^2 |\lambda_{\min}(P)|$. In addition, we have the decomposition $\Omega \oslash \widehat N  = \widetilde \Omega + (N\oslash \widehat N -{\bf 1}_n{\bf 1}_n') \circ \widetilde \Omega$. Since  $\| (N\oslash \widehat N -{\bf 1}_n{\bf 1}_n') \circ \widetilde \Omega \| \ll | \lambda_{K} (\widetilde \Omega) |$ with high probability, we obtain that 
\[
\lambda_k(\Omega \oslash \widehat N)  = \lambda_k(\widetilde \Omega) (1+ o(1)), \quad \text{ for $1\leq k \leq K$.} 
\]
Consequently, recall the definition that $\tau_n =   n\bar \theta^3 \| \widehat P - P\|_{\max}  +  \bar \theta^2\sqrt{n\bar \theta^2} \| \widehat \Pi - \Pi\| $, we obtain 
\begin{align} \label{2024092100}
&\min\{\| \hat \xi_1 - \xi_1\|, \|\Xi_1'\hat \xi_1\|\}\leq C\frac{  \| (N\oslash \widehat N -{\bf 1}_n{\bf 1}_n') \circ \widetilde \Omega \| +\tau_n  + \sqrt{\lambda_1(\widetilde \Omega) }}{\lambda_1(\widetilde \Omega)} = o(1) \notag\\
&\| \widehat \Xi_1 O- \Xi \|_F \leq C\frac{  \| (N\oslash \widehat N -{\bf 1}_n{\bf 1}_n') \circ \widetilde \Omega \| +\tau_n   + \sqrt{\lambda_1(\widetilde \Omega ) }}{\big| \lambda_K(\widetilde \Omega) \big|} = o(1)
\end{align}
with probability $1- o(n^{-3})$.

To proceed, we need to study the entry-wise error for $\hat \xi_1 - \xi_1$. By $(A\oslash \widehat N ) \hat \xi_1 = \hat \lambda_1 \hat \xi_1$ and $A\oslash \widehat N = \widetilde \Omega + \widetilde W$, we derive 
\begin{align*}
 \hat \xi_1 - \xi_1=( \hat \lambda_1^{-1} \lambda_1 \xi_1' \hat \xi_1 - 1)\xi_1 + \hat \lambda_1^{-1} \Xi_1 {\rm diag} (\lambda_2, \cdots, \lambda_K)  \Xi_1' \hat \xi_1 + \hat \lambda_1^{-1}\widetilde W \hat \xi_1\,.
\end{align*}
We  can bound 
\begin{align*}
| \hat \lambda_1^{-1} \lambda_1 \xi_1' \hat \xi_1 - 1| &\leq C (|\lambda_1^{-1}(\hat \lambda_1- \lambda_1 )| + |\xi_1'\hat \xi_1- 1|) \leq C\Big( \frac{\|\widetilde W\|}{\lambda_1(\widetilde \Omega)} +  \frac{\|\widetilde W\|^2}{\lambda_1(\widetilde \Omega)^2} \Big)  \notag\\
& \leq C\frac{  \| (N\oslash \widehat N -{\bf 1}_n{\bf 1}_n') \circ \widetilde \Omega \| +\tau_n  + \sqrt{\lambda_1(\widetilde \Omega) }}{\lambda_1(\widetilde \Omega)} 
\end{align*}
and 
\begin{align*}
 \| \hat \lambda_1^{-1}  {\rm diag} (\lambda_2, \cdots, \lambda_K)  \Xi_1' \hat \xi_1 \|\leq \|  \Xi_1' \hat \xi_1\| \leq C\frac{  \| (N\oslash \widehat N -{\bf 1}_n{\bf 1}_n') \circ \widetilde \Omega \| +\tau_n  + \sqrt{\lambda_1(\widetilde \Omega) }}{\lambda_1(\widetilde \Omega)} 
\end{align*}
These give rise to 
\begin{align}\label{2024092104}
|\hat \xi_1(i) - \xi_1(i) |\leq C\frac{  \| (N\oslash \widehat N -{\bf 1}_n{\bf 1}_n') \circ \widetilde \Omega \| +\tau_n  + \sqrt{\lambda_1(\widetilde \Omega) }}{\lambda_1(\widetilde \Omega)}   \cdot \frac{1}{\sqrt n} + \frac{|e_i' \widetilde W \hat \xi_1|}{n\bar \theta^2}
\end{align}
since $\|\Xi_1(i)\| \leq 1/\sqrt n$ and $\xi_1(i) \asymp 1/\sqrt n$ following from the assumptions on $\widetilde \Omega$ (see (\ref{bdd:eigenvector})) .

What remains to bound ${|e_i' \widetilde W \hat \xi_1|}/{n\bar \theta^2}$. Using the definition of $\widetilde W$, we first have 
\begin{align*}
& \quad {|e_i' \widetilde W \hat \xi_1|} \notag\\
&\leq  \big| e_i' [(N\oslash \widehat N -{\bf 1}_n {\bf 1}_n') \circ \widetilde \Omega]\hat \xi_1  \big| +  \big| (N\circ \widetilde \Omega \oslash \widehat N)_{ii} \hat \xi_1(i)   \big|+ \big| e_i' (W\oslash \widehat N )\hat \xi_1\big| \notag\\
& \leq  \|e_i' [(\widehat \Theta \widehat \Pi \widehat  P \widehat \Pi'  \widehat \Theta - \Theta \Pi P \Pi' \Theta) \circ \widetilde \Omega \circ N ]\|_1 \|\hat \xi_1\|_{\max}+ \bar \theta^2\Big(\frac{1}{\sqrt n}+ |\hat \xi_1(i) - \xi_1(i) | \Big) + \big| e_i ( W\oslash \widehat N) \hat \xi_1\big| \notag \\
& \leq n\bar \theta^4 \Big(\frac{1}{\sqrt n} +  \|\hat \xi_1 - \xi \|_{\max}\Big) + \bar \theta^2\Big(\frac{1}{\sqrt n}+ |\hat \xi_1(i) - \xi_1(i) | \Big) + \big| e_i ( W\oslash \widehat N) \hat \xi_1\big| 
\end{align*}
Here we crudely bound 
\[
\|e_i' [(N\oslash \widehat N -{\bf 1}_n {\bf 1}_n') \circ \widetilde \Omega]\|_1 \leq  \|e_i' [(\widehat \Theta \widehat \Pi \widehat  P \widehat \Pi'  \widehat \Theta - \Theta \Pi P \Pi' \Theta) \circ \widetilde \Omega \circ N ]\|_1 \leq C n\bar \theta^4\,. 
\]
and $\|N\circ \widetilde \Omega \oslash \widehat N \|_{\max} \leq C\| \widetilde \Omega\|_{\max}\leq C \bar \theta^2$ by the fact that all entries in $N$ and $\widehat N$ are lowered bound by a positive constant. 

%
%
Regarding  the last term on the RHS, i.e., $\big| e_i ( W\oslash \widehat N) \hat \xi_1\big| $, we have 
\begin{align*}
\big| e_i ( W\oslash \widehat N) \hat \xi_1\big| &\leq \big| e_i ( W\oslash  N) \hat \xi_1\big| + \big| e_i'W\oslash N \circ(N\oslash \widehat N - {\bf 1}_n {\bf 1}_n')\hat \xi_1\big| 
\end{align*}
We study the second term below. Note that $W\oslash N \circ(N\oslash \widehat N - {\bf 1}_n {\bf 1}_n') = W\circ (\widehat \Theta \widehat \Pi \widehat P \widehat \Pi' \widehat \Theta  - \Theta \Pi P \Pi'\Theta)$. We bound
\begin{align*}
\big| e_i'W\oslash N \circ(N\oslash \widehat N - {\bf 1}_n {\bf 1}_n')\hat \xi_1\big| &\leq 
 | e_i'( W\circ \widehat \Theta \widehat \Pi (\widehat P - P)  \widehat \Pi' \widehat \Theta) \hat \xi_1 |
+ | e_i'( W\circ  \widehat \Theta (\widehat \Pi - \Pi) P \widehat \Pi'\widehat \Theta)  \hat \xi_1 | \notag\\
& \quad  + | e_i'(W\circ\widehat \Theta \Pi P (\widehat \Pi - \Pi)'\widehat \Theta ) \hat \xi_1 |  + | e_i'( W\circ  ( \widehat \Theta - \Theta)  \Pi P  \Pi'\widehat \Theta ) \hat \xi_1 |  \notag\\
& \quad +| e_i'(  W\circ \Theta \Pi P \Pi'( \widehat \Theta - \Theta) ) \hat \xi_1 |\notag\\
\end{align*}
For each term, we further have 
\begin{align*}
 & | e_i'( W\circ \widehat \Theta \widehat \Pi (\widehat P - P)  \widehat \Pi' \widehat \Theta) \hat \xi_1 |\leq  | e_i' \widehat \Theta ( W\circ \widehat \Pi (\widehat P - P)  \widehat \Pi' ) \widehat \Theta\hat \xi_1 | \leq \bar \theta^2 \|e_i'W\circ \widehat \Pi (\widehat P - P)  \widehat \Pi'  \|_1 \|\hat \xi_1\|_{\max}\notag\\
 & \qquad \qquad \qquad  \qquad \qquad \qquad \leq \bar \theta^2\|e_i'W\|\|e_i' \widehat \Pi (\widehat P - P)  \widehat \Pi' \|\|\hat \xi_1\|_{\max}  \leq n\bar \theta^3 \|\widehat P - P \|_{\max}  \|\hat \xi_1\|_{\max} \notag\\
&  | e_i'( W\circ  \widehat \Theta (\widehat \Pi - \Pi) P \widehat \Pi'\widehat \Theta)  \hat \xi_1 | \leq C \bar \theta \|e_i'W\| \|e_i'  (\widehat \Pi - \Pi) P \widehat \Pi'\widehat \Theta) \circ  \hat \xi_1\| \leq  C \bar \theta^2\|e_i'W\| \|(\widehat \Pi - \Pi) P \widehat \Pi'\|_{\max} \|\hat \xi_1\| 
\notag\\
& \qquad \qquad \qquad  \qquad \qquad \qquad  \leq   C \bar \theta^2\|e_i'W\| \leq C\bar \theta^2 \sqrt{n\bar \theta^2} \notag\\
 &  | e_i'(W\circ\widehat \Theta \Pi P (\widehat \Pi - \Pi)'\widehat \Theta ) \hat \xi_1 |\leq C \bar \theta \|e_iW\| \|e_i'  \Pi P (\widehat \Pi - \Pi)'\widehat \Theta) \circ  \hat \xi_1\|
 \leq   C \bar \theta^2\|e_iW\| \leq C\bar \theta^2 \sqrt{n\bar \theta^2} \notag\\
 &  | e_i'( W\circ  ( \widehat \Theta - \Theta)  \Pi P  \Pi'\widehat \Theta ) \hat \xi_1 | \leq |\hat \theta_i - \theta_i |  | e_i'( W\circ   \Pi P  \Pi') \widehat \Theta \hat \xi_1 | \leq C\bar \theta^2 \|e_i'( W\circ   \Pi P  \Pi')\|\leq C\bar \theta^2 \sqrt{n\bar \theta^2} \notag\\
  & | e_i'(  W\circ \Theta \Pi P \Pi'( \widehat \Theta - \Theta) ) \hat \xi_1 | \leq C\theta_i   | e_i'( W\circ   \Pi P  \Pi') ( \widehat \Theta - \Theta) \hat \xi_1 | \leq C\bar \theta^2 \|e_i'( W\circ   \Pi P  \Pi')\|\leq C\bar \theta^2 \sqrt{n\bar \theta^2}
\end{align*}
Combining all these inequalities, we arrive at 
\begin{align}\label{2024092105}
\frac{|e_i' \widetilde W \hat \xi_1|}{n\bar \theta^2} &\leq\frac Cn \Big(\frac{1}{\sqrt n}+ |\hat \xi_1(i) - \xi_1(i) | \Big) +C(\bar \theta^2 + \bar \theta\|\widehat P - P\|_{\max} )\Big(\frac{1}{\sqrt n} +  \|\hat \xi_1 - \xi_1 \|_{\max}\Big) \notag\\
& \qquad +\frac{C\bar \theta}{\sqrt n}+ \frac{\big| e_i' ( W\oslash  N) \hat \xi_1\big|}{n\bar \theta^2}
\end{align}

In the sequel, we analyze $\big| e_i' ( W\oslash  N) \hat \xi_1\big|$ by leave-one-out technique. Let $\hat \xi_1^{(i)}$ be the first eigenvector of 
\[
A^{(i)} \oslash N = \widetilde \Omega  - {\rm diag} ( \widetilde \Omega) + W^{(i)}\oslash N
\]
where $W^{(i)}$ is obtained by zeroing out the $i$-th row and column of $W$.
Without loss of generality, we assume ${\rm sgn}(\hat \xi_1' \hat \xi_1^{(i)}) = 1$. Thus, 
\begin{align*}
\big| e_i' ( W\oslash  N) \hat \xi_1\big| &\leq \big| e_i' ( W\oslash  N) \hat \xi_1^{(i)}\big| + \|e_i ( W\oslash  N)\| \| \hat \xi_1 -  \hat \xi_1^{(i)} \| \notag\\
& \leq  C\Big( \bar \theta \sqrt{\log(n)} + \|\hat \xi_1^{(i)}\|_{\max} \log(n) + \sqrt{n\bar \theta^2} \| \hat \xi_1 -  \hat \xi_1^{(i)} \| \Big) \notag\\
& \leq C\Big( \bar \theta \sqrt{\log(n)} + \|\hat \xi_1\|_{\max} \log(n) + \sqrt{n\bar \theta^2} \| \hat \xi_1 -  \hat \xi_1^{(i)} \| \Big)
\end{align*}
where we applied Bernstein inequality on $e_i' ( W\oslash  N) \hat \xi_1^{(i)}$ as $e_i' ( W\oslash  N) $ is independent of $\hat \xi_1^{(i)}$. And the last step is due to the derivation 
\[
\|\hat \xi_1^{(i)}\|_{\max} \log(n) \leq \big(\|\hat \xi_1\|_{\max}+ \|\hat \xi_1^{(i)} -\hat \xi_1\|_{\max} \big) \log(n) \leq \|\hat \xi_1\|_{\max} \log(n) + \|\hat \xi_1^{(i)} -\hat \xi_1\| \sqrt{n\bar \theta^2} 
\]
under the condition that $\sqrt{n\bar \theta^2}\lambda_{\min} (P) \geq c_3\log(n)$. Next, by sine-theta theorem, we bound 
\begin{align}\label{2024092101}
 \| \hat \xi_1 -  \hat \xi_1^{(i)} \| &\leq \frac{\| \big(\widetilde W - W^{(i)}\oslash N +  {\rm diag} ( \widetilde \Omega) \big) \hat \xi_1\|}{ n\bar \theta^2} \notag\\
& \leq \frac{\| e_ie_i'( W\oslash  N) \hat \xi_1+ ( W\oslash  N) e_ie_i' \hat \xi_1\|}{n\bar \theta^2} + \frac{\|  \big((N\oslash \widehat N -{\bf 1}_n {\bf 1}_n') \circ \widetilde \Omega\big)  \hat \xi_1\|}{n\bar \theta^2} \notag\\
& \qquad +  \frac{\|W\oslash N \circ(N\oslash \widehat N - {\bf 1}_n {\bf 1}_n') \hat \xi_1 \|}{n\bar \theta^2}+ \frac{C}{n}
\end{align}
where we used the decomposition
\[
\widetilde W = (N\oslash \widehat N -{\bf 1}_n {\bf 1}_n') \circ \widetilde \Omega - {\rm diag} (N\circ \widetilde \Omega \oslash \widehat N) + W\oslash  N +  W\oslash N \circ(N\oslash \widehat N - {\bf 1}_n {\bf 1}_n')
\]
and the crude bound
\begin{align*}
& \big\| \big[{\rm diag} ( \widetilde \Omega) - {\rm diag} (N\circ \widetilde \Omega \oslash \widehat N) \big] \hat \xi_1 \big\| \leq C\bar \theta^2 \|\hat \xi_1\| \leq C\bar \theta^2 
\end{align*}
following from $N\oslash \widehat N \leq C$ with high probability and $\|\widetilde \Omega\|_{\max} \leq C \bar \theta^2$. To proceed, we further  bound
\begin{align}\label{2024092102}
\frac{\| \big(  (N\oslash \widehat N -{\bf 1}_n {\bf 1}_n') \circ \widetilde \Omega \big) \hat \xi_1\|}{n\bar \theta^2} \leq \frac{\|  (N\oslash \widehat N -{\bf 1}_n {\bf 1}_n') \circ \widetilde \Omega\|}{n\bar \theta^2} \leq \frac{\|  N\oslash \widehat N -{\bf 1}_n {\bf 1}_n' \|_F}{n}
\end{align}
And we analyze the upper bound for $\|W\oslash N \circ(N\oslash \widehat N - {\bf 1}_n {\bf 1}_n') \hat \xi_1 \|$ below. By definition, 
\begin{align*}
\|W\oslash N \circ(N\oslash \widehat N - {\bf 1}_n {\bf 1}_n') \hat \xi_1 \| & = \sqrt{\sum_i \Big(\sum_j (W\oslash N)_{ij}(N\oslash \widehat N - {\bf 1}_n {\bf 1}_n')_{ij} \hat \xi_1(j)  \Big)^2} \notag\\
& \leq  \sqrt{\sum_i  \sum_j (N\oslash \widehat N - {\bf 1}_n {\bf 1}_n')_{ij}^2 \cdot \sum_j (W\oslash N)_{ij}^2 \hat \xi_1(j)^2} \notag\\
& \leq  \sqrt{\sum_{i,j} (N\oslash \widehat N - {\bf 1}_n {\bf 1}_n')_{ij}^2 \cdot \max_i \sum_j (W\oslash N)_{ij}^2 \|\hat \xi_1\|^2_{\max}} \notag\\
& \leq \|N\oslash \widehat N - {\bf 1}_n {\bf 1}_n'\|_F \|\hat \xi_1\|_{\max} \sqrt{\max_i \sum_j (W\oslash N)_{ij}^2}
\end{align*}
where in the second step we used Cauchy-Schwarz inequality. 
Regarding the last factor inside the square root, for each fixed $i$, it is a sum of independent r.v.s, so we can use Bernstein inequality to get its high probability bound. Specifically, fixed an $i$, as $N$ is deterministic and each entry of $N$  is $\asymp 1$, we can derive the mean of $\sum_j (W\oslash N)_{ij}^2$  is given by 
\[
\sum_j \mathbb E(W\oslash N)_{ij}^2\asymp \sum_j \theta_i \theta_j \asymp n\bar \theta^2\, ;
\]
And the variance can be estimated by 
\[
\sum_j {\rm var} (W\oslash N)_{ij}^2 \leq \sum_j \mathbb E(W\oslash N)_{ij}^4 \leq C n\bar \theta^2\,. 
\]
Consequently, by Bernstein inequality, it is not hard to derive 
\[
\Big| \sum_j (W\oslash N)_{ij}^2 - \mathbb E(W\oslash N)_{ij}^2 \Big| \leq C\sqrt{n\bar \theta^2 \log(n) } +C \log(n) \ll C n\bar \theta^2
\] 
with probability $1- o(n^{-4})$. Then, combining all $i$, it gives that 
\[
\max_i \sum_j (W\oslash N)_{ij}^2\leq C n\bar \theta^2
\]
with probability $1- o(n^{-3})$. We thus obtain that 
\begin{align}\label{2024092301}
\|W\oslash N \circ(N\oslash \widehat N - {\bf 1}_n {\bf 1}_n') \hat \xi_1 \|\leq C\|N\oslash \widehat N - {\bf 1}_n {\bf 1}_n'\|_F \cdot  \sqrt n \bar \theta \|\hat \xi_1\|_{\max}
\end{align}

Next, we study the bound of $\| e_ie_i'( W\oslash  N) \hat \xi_1+ ( W\oslash  N) e_ie_i' \hat \xi_1\|$ below. 
\begin{align}\label{2024092103}
& \quad \| e_ie_i'( W\oslash  N) \hat \xi_1+ ( W\oslash  N) e_ie_i' \hat \xi_1\| \notag\\
& \leq |e_i'( W\oslash  N)  \hat \xi_1 | + \| e_i'( W\oslash  N) \| |\hat \xi_1(i) | \notag\\
& \leq  C\Big( \bar \theta \sqrt{\log(n)} + \|\hat \xi_1\|_{\max} \log(n) + \sqrt{n\bar \theta^2} \| \hat \xi_1 -  \hat \xi_1^{(i)} \| + \sqrt{n\bar \theta^2} |\hat \xi_1(i) | \Big).
\end{align}

Combining (\ref{2024092101}) - (\ref{2024092103}), we get  
\begin{align*}
\| \hat \xi_1 -  \hat \xi_1^{(i)} \| &  \leq C\bigg( \frac{ \sqrt{\log(n)} }{n\bar \theta}+ \|\hat \xi_1\|_{\max}\frac{ \log(n) }{n\bar \theta^2} +\frac{ \| \hat \xi_1 -  \hat \xi_1^{(i)} \|}{ \sqrt{n\bar \theta^2}}   \notag\\
& \qquad +\frac{ |\hat \xi_1(i) |}{ \sqrt{n\bar \theta^2}} + \frac{\|\hat \xi_1\|_{\max} \|  N\oslash \widehat N -{\bf 1}_n {\bf 1}_n' \|_F}{\sqrt{n\bar \theta^2}} \bigg)
\end{align*}
Rearranging both sides gives that 
\[
\| \hat \xi_1 -  \hat \xi_1^{(i)} \| \leq C\bigg( \frac{ \sqrt{\log(n)} }{n\bar \theta}+ \|\hat \xi_1\|_{\max}\frac{1 + \|N\oslash \widehat N - {\bf 1}_n{\bf 1}_n' \|_{F} }{\sqrt{n\bar \theta^2}}   \bigg)
\]
Consequently, 
\begin{align*}
\frac{\big| e_i ( W\oslash  N) \hat \xi_1\big|}{n\bar \theta^2} &\leq C\bigg(  \frac{ \sqrt{\log(n)} }{n\bar \theta}+ \frac{\|\hat \xi_1\|_{\max} \log(n)}{n\bar \theta^2} + \frac{ \| \hat \xi_1 -  \hat \xi_1^{(i)} \|}{\sqrt{n\bar \theta^2}} \bigg) \notag\\
& \leq C\bigg(  \frac{ \sqrt{\log(n)} }{n\bar \theta} + \|\hat \xi_1\|_{\max}\Big[ \frac{ \log(n)}{n\bar \theta^2} + \frac{ \|N\oslash \widehat N - {\bf 1}_n{\bf 1}_n' \|_{F}}{n\bar \theta^2} \Big]\bigg)
\end{align*}
Plugging this and (\ref{2024092105}) into (\ref{2024092104}), we have 
\begin{align*}
|\hat \xi_1(i) - \xi_1(i) |& \leq C\frac{  \| (N\oslash \widehat N -{\bf 1}_n{\bf 1}_n') \circ \widetilde \Omega \| +\tau_n  + \sqrt{\lambda_1(\widetilde \Omega) }}{\lambda_1(\widetilde \Omega)}   \cdot \frac{1}{\sqrt n} + \frac{C\bar \theta}{\sqrt n}\notag\\
&\quad  + C(\bar \theta^2 + \bar \theta \|\widehat P - P\|_{\max} )\Big(\frac{1}{\sqrt n} +  \|\hat \xi_1 - \xi \|_{\max}\Big)+ \frac Cn  |\hat \xi_1(i) - \xi_1(i) | \notag\\
& \quad +C\bigg(  \frac{ \sqrt{\log(n)} }{n\bar \theta}+ \|\hat \xi_1\|_{\max}\Big[ \frac{ \log(n)}{n\bar \theta^2} + \frac{ \|N\oslash \widehat N - {\bf 1}_n{\bf 1}_n' \|_{F}}{n\bar \theta^2} \Big]  \bigg) 
\end{align*}
Rearranging both sides, together with $\|\hat \xi_1\|_{\max}  \leq \| \xi_1\|_{\max} + \|\hat \xi_1 - \xi_1\|_{\max}\leq C/\sqrt n + \|\hat \xi_1 - \xi_1\|_{\max}$, gives rise to  
\begin{align*}
& \quad |\hat \xi_1(i) - \xi_1(i) |\notag\\
& \leq C \Big( \frac{ \| (N\oslash \widehat N -{\bf 1}_n{\bf 1}_n') \circ \widetilde \Omega \|+ \tau_n }{n\bar \theta^2}  +  \frac{ \sqrt{\log(n)} }{\sqrt{n\bar \theta^2} } + \frac{\|N\oslash \widehat N -{\bf 1}_n {\bf 1}_n'\|_{F} }{n\bar\theta^2 }  + \bar \theta +\bar \theta \|\widehat P - P\|_{\max}\Big) \frac{1}{\sqrt{n}}\notag\\
& \quad +C\Big(  \frac{\|N\oslash \widehat N -{\bf 1}_n {\bf 1}_n'\|_{F} }{n\bar \theta^2} + \bar \theta^2 + \bar \theta\|\widehat P - P\|_{\max}+ \frac{\log(n)}{n\bar\theta^2} \Big) \|\hat \xi_1 - \xi_1\|_{\max}
\end{align*}
We further take maximum for both sides. Under the conditions that  $\| (N\oslash \widehat N -{\bf 1}_n{\bf 1}_n') \circ \widetilde \Omega \|\ll n\bar \theta^2$, $\|N\oslash \widehat N -{\bf 1}_n {\bf 1}_n'\|_F \ll n\bar \theta^2$ and $n\bar \theta^2 \geq C(\log(n))^2$, together with $\bar \theta \|\widehat P - P\|_{\max} \ll \lambda_{\min}(P) \leq C$ and $\bar \theta = o(1)$,  it yields that with probability $1- o(n^{-3})$,
{
\begin{align*}
 \|\hat \xi_1 - \xi_1\|_{\max} &\leq  \frac{C}{\sqrt{n}} \Big( \frac{ \| (N\oslash \widehat N -{\bf 1}_n{\bf 1}_n') \circ \widetilde \Omega \|+ \tau_n }{n\bar \theta^2}  +  \frac{ \sqrt{\log(n)} }{\sqrt{n\bar \theta^2} } + \frac{\|N\oslash \widehat N -{\bf 1}_n {\bf 1}_n'\|_{F} }{n\bar\theta^2 }  \notag\\
 & \qquad \quad  + \bar \theta^2 + \bar \theta \|\widehat P - P\|_{\max}\Big)\notag\\
 &  \ll \frac{1}{\sqrt n}
\end{align*}
}
Further by $|\xi_1(i)| \asymp 1/\sqrt n$ for $1\leq i \leq n$, we deduce that $|\hat \xi_1(i)| \asymp 1/\sqrt n$ for $1\leq i \leq n$. 

Now, by the definition of $\widehat R$, we can derive  
\begin{align*}
\|\widehat RO - R \|_F^2 &=\|  {\rm diag}(\hat \xi_1)^{-1} (\widehat \Xi_1O - \Xi) -  \big[ {\rm diag}(\hat \xi_1)^{-1}  - {\rm diag}( \xi_1)^{-1} \big] \Xi_1\|^2_F \notag\\
& \leq C\Big(\|  {\rm diag}(\hat \xi_1)^{-1} (\widehat \Xi_1O - \Xi) \|_F^2 + \|  \big[ {\rm diag}(\hat \xi_1)^{-1}  - {\rm diag}( \xi_1)^{-1} \big] \Xi_1\|_F^2\big) \notag\\
& \leq C\Big( { n }\|\widehat \Xi_1O - \Xi\|_F^2 + \|\hat \xi_1 - \xi_1\|^2\|{\rm diag}(\hat \xi_1)^{-1}   {\rm diag}( \xi_1)^{-1}\Xi_1\|_{2\to \infty}^2 \Big) \notag\\
& \leq C\Big({ n }\|\widehat \Xi_1O - \Xi\|_F^2 + n\, \|\hat \xi_1 - \xi_1\|^2 \Big) 
\end{align*}
By (\ref{2024092100}), we conclude that 
\begin{align*}
\|\widehat RO - R \|_F^2 & \leq  Cn\frac{  \| (N\oslash \widehat N -{\bf 1}_n{\bf 1}_n') \circ \Omega  \|^2  +\tau_n^2   +{ \lambda_1(\widetilde \Omega) } }{\big| \lambda_K(\widetilde \Omega)  \big|^2}
\end{align*}
with probability $1- o(n^{-3})$. Now, combining the above result with  Lemma~\ref{lem: Hamming-R}, we conclude the proof of Lemma~\ref{lemma:RSCORE}.

\subsection{The error rate of $\theta$} \label{subsec:refit_theta} 
In this subsection, we prove the error rate for refitting $\theta$ under the assumptions in Theorem~\ref{thm:main}.The results is collected in the following Lemma. 
\begin{lemma} \label{lem:error_theta}
Under the assumptions in Theorem~\ref{thm:main}, it holds that  with probability $1 - o(n^{-3})$, 
\begin{align*}
\begin{array}{ll}
 \big| \hat \theta_i - \theta_i  \big| \leq C\Big( \sqrt{\log(n)/n} + r_n/\bar\theta\Big), & \text{ if  \quad $\hat \pi_i = \pi_i$}; \\
 \big| \hat \theta_i - P_{k k_0}\theta_i  \big| \leq C\Big( \bar\theta (\log(n) /n\bar \theta^2)^{1/4}+ \sqrt{r_n}\Big), & \text{ if \quad $\hat \pi_i = e_{k} \neq e_{k_0} =  \pi_i$}.
 \end{array}
\end{align*}
 where  $r_n$ is the Hamming error of the $\widehat \Pi$ by directly SCORE  and $\{e_k\}_{k=1}^K $  represents the standard basis of $\mathbb R^{K}$.
\end{lemma}

We prove Lemma~\ref{lem:error_theta} below. 

Recall the refitting formula for $\theta$:
\begin{align}\label{eq:express_htheta}
\hat{\theta}_i = \sqrt{ \frac{\sum_{j \neq t \in \hat{S}_{k, i}}  A_{ij} (1- A_{jt})  A_{ti}}{\sum_{j \neq t \in \hat{S}_{k,i}}  (1- A_{ij})  A_{jt} (1- A_{ti})}}
\end{align}
where $i \in \widehat{ \mathcal C}_k $ and $\hat{S}_{k, i} = \widehat{\mathcal C}_k \setminus \{i\}$.
Using the error rate of $\widehat \Pi$ from SCORE, i.e., $\|\widehat \Pi - \Pi\|_1 \leq nr_n$, we first crudely bound the numerator and denominator in the expression of $\hat \theta_i$. Let $\hat 1_k$ and $1_k$ denote the $k$0th column of $\widehat \Pi $ and $\Pi$, respectively.
\begin{align*}
&\quad \sum_{j \neq t \in \hat{S}_{k, i}}  A_{ij} (1- A_{jt})  A_{ti}  \notag\\
 & = e_{n,i}' A {\rm diag} \big(\hat 1_k \big)({\bf 1}_n {\bf 1}_n' -I_n- A) {\rm diag} \big(\hat 1_k \big)A e_{n,i} \notag\\
& =  e_{n,i}' A {\rm diag} \big( 1_k \big)({\bf 1}_n {\bf 1}_n' -I_n- A) {\rm diag} \big(1_k \big)A e_{n,i} + e_{n,i}' A {\rm diag} \big( \hat 1_k - 1_k \big)({\bf 1}_n {\bf 1}_n' -I_n- A) {\rm diag} \big(1_k \big)A e_{n,i} \notag\\
& \quad + e_{n,i}' A {\rm diag} \big( \hat 1_k \big)({\bf 1}_n {\bf 1}_n' -I_n- A) {\rm diag} \big(\hat 1_k - 1_k \big)A e_{n,i} 
\end{align*}
For the second and third terms on the RHS above,  we bound 
\begin{align*}
& \big| e_{n,i}' A {\rm diag} \big( \hat 1_k - 1_k \big)({\bf 1}_n {\bf 1}_n' -I_n- A) {\rm diag} \big(1_k \big)A e_{n,i}  + e_{n,i}' A {\rm diag} \big( \hat 1_k \big)({\bf 1}_n {\bf 1}_n' -I_n- A) {\rm diag} \big(\hat 1_k - 1_k \big)A e_{n,i} \big| \notag\\
& \leq 2 \|e_{n,i}' A \|_1 \cdot  \| e_{n,i}' A {\rm diag} \big( \hat 1_k - 1_k \big)\|_1\notag\\
& \leq Cn\bar \theta^2 \| \hat 1_k - 1_k\|_1 \leq Cn\bar \theta^2 \| \widehat \Pi - \Pi \|_1 \leq Cn^{2}r_n\bar \theta^2 
\end{align*}
where to obtain the third line, we used 
\begin{align*}
& \|e_{n,i}' A \|_1  = \sum_{j\neq i} A_{ij} =  \sum_{j\neq i} \Omega_{ij} +  \sum_{j\neq i} W_{ij}  = n\bar \theta^2 + O(\sqrt{n\bar \theta^2 \log(n)}\, )\notag\\
&  \| e_{n,i}' A {\rm diag} \big( \hat 1_k - 1_k \big)\|_1\leq \|e_{n,i}' A\|_{\max} \| \hat 1_k - 1_k\|_1 \leq \| \hat 1_k - 1_k\|_1
\end{align*}
simultaneously for all $1\leq i \leq n$, with probability $1- o(n^{-3})$. Here Bernstein inequality is employed to derive 
\[
\Big|  \sum_{j\neq i} W_{ij} \Big| \leq C\big(\sqrt{ \sum_{j\neq i}{\rm var} (W_{ij}) \log(n) } + \log(n) \big) 
\leq C\sqrt{n\bar \theta^2 \log(n)}
\]
by noting that $W_{ij}$ is with mean $0$ and variance 
\[
{\rm var} (W_{ij})  = \Omega_{ij} (1- \Omega_{ij}) \leq \theta_i\theta_j
\]
and the condition $n\bar \theta^2 \geq C \log(n)$. 
Therefore, 
\begin{align}\label{error:theta_fromPi1}
\Big| \sum_{j \neq t \in \hat{S}_{k, i}}  A_{ij} (1- A_{jt})  A_{ti}  - \sum_{j \neq t \in {S}_{k, i}}  A_{ij} (1- A_{jt})  A_{ti}  \Big| \leq C  n^{2} r_n \bar \theta^2 \,.
\end{align}
simultaneously for all $1\leq i \leq n$, with probability $1- o(n^{-3})$. 
Similarly, we can show that 
 \begin{align}\label{error:theta_fromPi2}
& \Big| \sum_{j \neq t \in \hat{S}_{k,i}}  (1- A_{ij})  A_{jt} (1- A_{ti})  -  \sum_{j \neq t \in {S}_{k,i}}  (1- A_{ij})  A_{jt} (1- A_{ti})\Big|  \notag\\
&
\leq C\|\hat 1_k - 1_k\|_1\cdot n\bar \theta^2 \leq C \|\hat \Pi -\Pi\|_1 n \bar \theta^2 \leq n^{2}r_n \bar \theta^2\,. 
  \end{align}
To proceed, we  study $ \sum_{j \neq t \in {S}_{k, i}}  A_{ij} (1- A_{jt})  A_{ti} $ and $ \sum_{j \neq t \in {S}_{k,i}}  (1- A_{ij})  A_{jt} (1- A_{ti})$ instead. Recall the decomposition $A = \Omega - {\rm diag}(\Omega) - W$. For the numerator,
%
 we decompose 
\begin{align*}
&  \sum_{j \neq t \in {S}_{k, i}}  A_{ij} (1- A_{jt})  A_{ti} \notag\\
&= \sum_{j \neq t \in {S}_{k, i}}  \Omega_{ij} (1- \Omega_{jt})  \Omega_{ti}  + \sum_{j \neq t \in {S}_{k, i}} W_{ij}(1- \Omega_{jt})  \Omega_{ti} + \Omega_{ij} (- W_{jt})  \Omega_{ti} + \Omega_{ij} (1- \Omega_{jt})  W_{ti} \notag\\
& \quad +  \sum_{j \neq t \in {S}_{k, i}} W_{ij}(- W_{jt})  \Omega_{ti} + \Omega_{ij} (- W_{jt})  W_{ti} + W_{ij} (1- \Omega_{jt})  W_{ti} \notag\\
 & \quad + \sum_{j \neq t \in {S}_{k, i}} W_{ij}(- W_{jt})  W_{ti} \notag\\
& =: \sum_{j \neq t \in {S}_{k, i}}  \Omega_{ij} (1- \Omega_{jt})  \Omega_{ti} + \sum_{a=1}^3 T_{1a}+ \sum_{a=1}^3T_{2a} + T_3 
\end{align*}
We analyze each term on the RHS above one by one. Note that $W_{ij}$ is with mean $0$ and variance 
\[
{\rm var} (W_{ij})  = \Omega_{ij} (1- \Omega_{ij}) \leq \theta_i\theta_j
\]
and the trivial bound $|W_{ij}\sum_{t\neq j}(1- \Omega_{jt})  \Omega_{ti} |<n\bar \theta^2$ (similarly for each summand in $T_{12}$ and $T_{13}$).
By Bernstein inequality, 
\begin{align*}
& |T_{11}|  \leq C\Big(\sqrt{\sum_{j} \theta_i\theta_j \big(\sum_{t\neq j} (1- \Omega_{jt}) \Omega_{ti} \big)^2 \log(n)} + n\bar \theta^2\log(n)\Big)  \leq  Cn\bar \theta^2\big(  \sqrt {n\bar \theta^2\log (n)}  + \log(n) \big) \notag\\ 
&|T_{12}|   \leq 2C\Big(\sqrt{\sum_{j< t} \theta_j\theta_t(\Omega_{ij})^2( \Omega_{ti})^2 \log(n)} + \log(n)\Big) \leq  C\big( n\bar \theta^5 \sqrt {\log (n)}  + \log(n) \big) \notag\\
& |T_{13}| \leq C\Big(\sqrt{\sum_{t} \theta_t\theta_i \big(\sum_{j\neq t} \Omega_{ij}(1- \Omega_{jt})\big) ^2\log(n)} + n\bar \theta^2\log(n)\Big)  \leq  Cn\bar \theta^2\big(  \sqrt {n\bar \theta^2\log (n)}  + \log(n) \big) 
\end{align*}
simultaneously for all $1\leq i \leq n$, with probability $1- o(n^{-3})$. Consequently, 
\begin{align*}
\Big|  \sum_{a=1}^3 T_{1a}\Big|  \leq  Cn\bar \theta^2  \sqrt {n\bar \theta^2\log (n)}  \,. 
\end{align*}
following from the condition that $n\bar \theta^2\geq C \log(n)$, which is implied by condition (\ref{asm:P}) in the manuscript.

Regarding $T_{2a}$ for $a=1,2,3$ and $T_3$, their large deviation bounds can be tackled by decoupling inequality for U-statistics in \cite{de1995decoupling}. Specifically, implied by this technique, the large deviation bound of $T_{21}$ is dominated by that of
\[
 \tilde T_{21}: = \sum_{j \neq t \in {S}_{k, i}} W_{ij}(- W^{(1)}_{jt})  \Omega_{ti} 
\]
where $W^{(1)}$ is an i.i.d. copy of $W$. Thanks to this independence, we first condition on $W^{(1)}$ and use Bernstein inequality to get 
\begin{align*}
| \tilde T_{21} | \leq C\Big[\sqrt{ \sum_{j \in {S}_{k, i}} \theta_i\theta_j\big(\sum_{t \neq j\in {S}_{k, i}} \Omega_{ti}W^{(1)}_{jt}\big)^2 \log(n)} + \max_{j} \Big| \sum_{t \neq j\in {S}_{k, i}} \Omega_{ti}W^{(1)}_{jt}\Big| \log(n)\Big]
\end{align*}
Next, by Bernstein inequality again, we obtain 
\[
\Big|\sum_{t \neq j\in {S}_{k, i}} \Omega_{ti}W^{(1)}_{jt} \Big| \leq C\Big(\sqrt{n\bar\theta^6 \log(n)} + \bar \theta^2 \log(n)\Big)
\]
Combining the above inequalities, we arrive at 
\[
| \tilde T_{21} | \leq C\big( n\bar \theta^4 {\log (n)}  +\sqrt{n\bar\theta^2} \bar \theta^2 (\log(n))^{3/2} \big) \,.
\]
simultaneously for all $1\leq i \leq n$, with probability $1- o(n^{-3})$. Therefore, 
\[
| T_{21} | \leq C\big( n\bar \theta^4 {\log (n)}  +\sqrt{n\bar\theta^2} \bar \theta^2  (\log(n))^{3/2} \big) \,.
\]
In the same manner, we can show that 
\begin{align*}
| T_{22} | \leq C\big( n\bar \theta^4 {\log (n)}  +\sqrt{n\bar\theta^2} \bar \theta^2 (\log(n))^{3/2}  \big), \qquad | T_{23} | \leq C n\bar \theta^2 {\log (n)} 
\end{align*}
As a result,
\begin{align*}
\Big|  \sum_{a=1}^3 T_{2a}\Big|  \leq   Cn\bar \theta^2 {\log (n)}   \,. 
\end{align*}
under the condition $n\bar \theta^2 \geq C \log(n) $ and $\bar \theta \leq C$.

Lastly, we prove the large deviation bound for $T_3$. Using decoupling inequality for U-statistics, it suffices to prove a large deviation bound for $\tilde T_3$ with 
\[
\tilde T_3: =  \sum_{j \neq t \in {S}_{k, i}} W_{ij}(- W^{(1)}_{jt})  W^{(2)}_{ti}
\]
Here $W^{(1)}$ and $W^{(2)}$ are two i.i.d. copies of $W$. Condition on $W^{(1)}, W^{(2)}$, by Bernstein inequality, 
\[
|\tilde T_3| \leq C\Big(\sqrt{\sum_{j\in {S}_{k, i}} \theta_i \theta_j \big( \sum_{t\neq j\in {S}_{k, i}} W^{(1)}_{jt}  W^{(2)}_{ti} \big)^2 \log(n)  } + \max_j \Big|  \sum_{t\neq j\in {S}_{k, i}} W^{(1)}_{jt}  W^{(2)}_{ti}\Big|\log(n) \Big)
\]
In addition, for $\sum_{t\neq j\in {S}_{k, i}} W^{(1)}_{jt}  W^{(2)}_{ti} $, each summand is independent of each other. By Bernstein inequality,  we can similarly get 
\[
\Big| \sum_{t\neq j\in {S}_{k, i}} W^{(1)}_{jt}  W^{(2)}_{ti} \Big|  \leq C\Big(  \sqrt{n\bar \theta^4 \log(n)}+\log(n) \Big)
\]
Consequently, 
\[
|\tilde T_3| \leq C\Big( n\bar \theta^3 \log(n)+ \sqrt{n\bar \theta^2}(\log(n))^{3/2} \Big)
\]
This, by  $n\bar \theta^2 \geq C \log(n)$ and $\bar \theta \leq C$,  further  implies 
\[
|T_3 |  \leq C n\bar \theta^2 {\log (n)} 
\]
simultaneously for all $1\leq i \leq n$, with probability $1- o(n^{-3})$. 
Based on the large deviation bounds for $\sum_{a=1}^3 T_{1a}, \sum_{a=1}^3 T_{2a}$ and $T_3$, we therefore conclude that 
\begin{align}\label{eq:numer_case1}
 \sum_{j \neq t \in {S}_{k, i}}  A_{ij} (1- A_{jt})  A_{ti}  = \sum_{j \neq t \in {S}_{k, i}}  \Omega_{ij} (1- \Omega_{jt})  \Omega_{ti} + O_{p} \big( (n\bar \theta^2)^{3/2}\sqrt{\log (n)}\, \big)
 \end{align}
 simultaneously for all $i$, where the probability is $1- o(n^{-3})$.
 
 Next, we analyze the denominator  $ \sum_{j \neq t \in {S}_{k,i}}  (1- A_{ij})  A_{jt} (1- A_{ti})$ . Analogously, 
 we decompose 
\begin{align*}
&  \sum_{j \neq t \in {S}_{k,i}}  (1- A_{ij})  A_{jt} (1- A_{ti})  \notag\\
 & =  \sum_{j \neq t \in {S}_{k,i}}  (1- \Omega _{ij})  \Omega_{jt} (1- \Omega_{ti})\notag\\
  &\quad  + \sum_{j \neq t \in {S}_{k,i}}  (1- \Omega_{ij})  \Omega_{jt} (-W_{ti})  + (1- \Omega _{ij}) W_{jt} (1- \Omega_{ti}) + (- W_{ij})  \Omega_{jt} (1- \Omega_{ti}) \notag\\
  & \quad +  \sum_{j \neq t \in {S}_{k,i}}  (1- \Omega_{ij})  W_{jt} (-W_{ti})  + (-W _{ij}) W_{jt} (1- \Omega_{ti}) + (- W_{ij})  \Omega_{jt} (- W_{ti}) \notag\\
  & \quad +   \sum_{j \neq t \in {S}_{k,i}}  (-W_{ij})  W_{jt} (-W_{ti})\notag\\
  &=: \sum_{j \neq t \in {S}_{k,i}}  (1- \Omega _{ij})  \Omega_{jt} (1- \Omega_{ti}) + \sum_{a=1}^3 \mathcal T_{1a}+ \sum_{a=1}^3\mathcal T_{2a} - T_3 
\end{align*}
Similarly to $ \sum_{a=1}^3 T_{1a}$ and $ \sum_{a=1}^3 T_{2a} $, we can derive 
\begin{align*}
& |\mathcal T_{11} | \leq Cn\bar \theta^2 \sqrt{n\bar \theta^2 \log(n) }, \quad |\mathcal T_{12} | \leq Cn\bar \theta \sqrt{\log(n) }, \quad |\mathcal T_{13} | \leq Cn\bar \theta^2 \sqrt{n\bar \theta^2 \log(n) }\notag\\
&  |\mathcal T_{21} | \leq Cn\bar \theta^2 \log(n) , \quad |\mathcal T_{22} | \leq Cn\bar \theta^2 \log(n) , \quad |\mathcal T_{23} | \leq Cn\bar \theta^4  \log(n) 
\end{align*}
by Bernstein inequality and decoupling inequality. Since the details are rather similar, we omit the details. The above estimates, together with $|T_3 |  \leq C n\bar \theta^2 {\log (n)} $, give rise to 
\begin{align}\label{eq:denom_case1}
 \sum_{j \neq t \in {S}_{k,i}}  (1- A_{ij})  A_{jt} (1- A_{ti})  = \sum_{j \neq t \in {S}_{k,i}}  (1- \Omega _{ij})  \Omega_{jt} (1- \Omega_{ti}) + O_{p} \big([(n\bar \theta^2)^{3 /2} + n\bar \theta]\sqrt{\log (n)}\, \big)
\end{align}
simultaneously for all $1\leq i \leq n$, with probability $1- o(n^{-3})$.

Combining (\ref{error:theta_fromPi1}), (\ref{error:theta_fromPi2}), (\ref{eq:numer_case1}) and (\ref{eq:denom_case1}) into (\ref{eq:express_htheta}), we can further derive that 
\begin{align}\label{eq:2024051601}
\hat{\theta}_i & = \sqrt{ \frac{\sum_{j \neq t \in \hat{S}_{k, i}}  A_{ij} (1- A_{jt})  A_{ti}}{\sum_{j \neq t \in \hat{S}_{k,i}}  (1- A_{ij})  A_{jt} (1- A_{ti})}} \notag\\
& = \sqrt{ \frac{\sum_{j \neq t \in {S}_{k, i}}  \Omega_{ij} (1- \Omega_{jt})  \Omega_{ti} + O_{p} \big( (n\bar \theta^2)^{3/2}\sqrt{\log (n)} \, + n^{2}r_n \bar \theta^2 \big) }{\sum_{j \neq t \in {S}_{k,i}}  (1- \Omega _{ij})  \Omega_{jt} (1- \Omega_{ti}) + O_{p} \big( [(n\bar \theta^2)^{3 /2} + n\bar \theta] \sqrt{\log (n)}\, + n^{2}r_n\bar \theta^2 \big)}} 
\notag\\
& =  \sqrt{ \frac{\sum_{j \neq t \in {S}_{k, i}}  \Omega_{ij} (1- \Omega_{jt})  \Omega_{ti} }{\sum_{j \neq t \in {S}_{k,i}}  (1- \Omega _{ij})  \Omega_{jt} (1- \Omega_{ti}) } + O_p \big(\bar\theta \sqrt{\log(n)/n} + r_n \big)  } \,  
\end{align}
simultaneously for all $1\leq i \leq n$, where the high probability is at least $1- o(n^{-3})$. 
Here we used the crude estimate 
\[
\sum_{j \neq t \in {S}_{k,i}}  (1- \Omega _{ij})  \Omega_{jt} (1- \Omega_{ti}) \asymp (n\bar \theta)^2
\]
under the assumption that the number of nodes in each community is balanced and the diagonal entries of $P$ are one so that $\Omega_{jt} \propto \theta_j\theta_t$ if $j,t\in \mathcal C_k$. 

To proceed, we separate the analysis into two cases: (1) $\hat \pi = \pi_0 = e_k$; (2) $\hat \pi_i = e_{k} \neq e_{k_0} = \pi_i$. This is because the leading term in $\hat{\theta}_i$ may vary with the two different cases. 

For case (1), $i \in \mathcal C_k$, it follows that 
\begin{align*}
 \frac{\sum_{j \neq t \in {S}_{k, i}}  \Omega_{ij} (1- \Omega_{jt})  \Omega_{ti} }{\sum_{j \neq t \in {S}_{k,i}}  (1- \Omega _{ij})  \Omega_{jt} (1- \Omega_{ti}) } =\frac{\sum_{j \neq t \in {S}_{k, i}}  N_{ij} N_{jt}  N_{ti} \theta_i\theta_j \theta_t\theta_i }{\sum_{j \neq t \in {S}_{k,i}}  N_{ij} N_{jt}  N_{ti} \theta_j \theta_t} = \theta_i^2
\end{align*}
which is also claimed by Lemma~\ref{lemma:theta} . In light of this, further with the condition in Theorem~\ref{thm:main} that $r_n \ll\bar\theta^2$ (note that $r_n \asymp \delta_n)$,  we conclude that 
\begin{align}\label{error_random1} 
\hat{\theta}_i = \theta_i + O_p \big( \sqrt{\log(n)/n} +r_n/\bar\theta \big) \,. 
\end{align}
simultaneously for all $1\leq i \leq n$, where the high probability is at least $1- o(n^{-3})$. 

For case (2), $i \notin \mathcal C_k$. Therefore, $\Omega_{ij} = N_{ij} \theta_i\theta_j \cdot P_{k_0k}$ where $N_{ij} = (1+ \theta_i\theta_j \cdot P_{k_0k})^{-1} $, for all $j\in \mathcal C_k$. As a result,  
\begin{align*}
 \frac{\sum_{j \neq t \in {S}_{k, i}}  \Omega_{ij} (1- \Omega_{jt})  \Omega_{ti} }{\sum_{j \neq t \in {S}_{k,i}}  (1- \Omega _{ij})  \Omega_{jt} (1- \Omega_{ti}) } =\frac{\sum_{j \neq t \in {S}_{k, i}}  N_{ij} N_{jt}  N_{ti} \theta_i\theta_j \theta_t\theta_i \cdot P_{k_0k}^2}{\sum_{j \neq t \in {S}_{k,i}}  N_{ij} N_{jt}  N_{ti} \theta_j \theta_t} = \theta_i^2 P_{k_0k}^2
\end{align*}
Note that $P_{k_0k} = \hat\pi_i ' P \pi_i$ under this case and $\|P\|_{\max} \leq C$. We therefore conclude that 
\begin{align}\label{error_random2} 
\hat{\theta}_i & = e'_{n,i}(\widehat \Pi P \Pi')e_{n,i}\theta_i + O_p \Big(\min \Big\{ \bar\theta (\log(n) /n\bar \theta^2)^{1/4}+ \sqrt{r_n},  \frac{\sqrt{\log(n)/n} + r_n/\bar\theta }{ e'_{n,i}(\hat \Pi P \Pi')e_{n,i}}  \Big\} \Big) \notag\\
& = e'_{n,i}(\widehat \Pi P \Pi')e_{n,i}\theta_i + O_p \big( \bar\theta (\log(n) /n\bar \theta^2)^{1/4}+ \sqrt{r_n} \big) \,. 
\end{align}
simultaneously for all $1\leq i \leq n$, where the high probability is at least $1- o(n^{-3})$. 
By (\ref{error_random1}) and (\ref{error_random2}), we  complete the proof.

\subsection{The error rate of $P$} \label{subsec:refit_P} 
In this subsection, we prove the error rate of $P$, which is presented in the following lemma.

\begin{lemma}\label{lem:error_P}
Under the assumptions in Theorem~\ref{thm:main}, it hold with probability $1 - o(n^{-3})$ that 
\[
\|\widehat{P} - P\|_{\max} \leq C \Big( \sqrt{\frac{\log(n) }{n\bar \theta^2}} + \frac{r_n}{\bar \theta^2}\Big)
\]
where $r_n $ denotes  the Hamming error by directly applying SCORE.
\end{lemma}

Recalling the refitting formula of $P$, 
\begin{equation} \label{supp:P2} 
\widehat{P}_{k\ell} = \frac{ \sum_{i \in \widehat{{\cal C}}_k}  \sum_{j  \in {\widehat {\cal C}}_{\ell}}  A_{ij}} {\sum_{i \in \widehat{ {\cal C}}_k}  \sum_{j  \in {\widehat {\cal C}}_{\ell}}\hat{\theta}_i \hat{\theta}_j (1 - A_{ij})}
\end{equation} 
for $k\neq \ell$. We can rewrite it as 
\[
\widehat{P}_{k\ell} =  \frac{\hat 1_k'  A \hat 1_{\ell}}{ \hat 1_k' \widehat \Theta  ({\bf 1}_n{\bf 1}_n' - A) \widehat \Theta \hat 1_{\ell} } 
\]
For the numerator, we derive 
\begin{align}\label{est:P_num}
\hat 1_k'  A \hat 1_{\ell} &= 1_k'  \Omega  1_{\ell} + 1_k'  W  1_{\ell}  + (\hat 1_k - 1_k)'  A 1_{\ell} + 1_k'  A  (\hat 1_{\ell} - 1_{\ell}) +  (\hat 1_k - 1_k)' A  (\hat 1_{\ell} - 1_{\ell})\notag\\
& = 1_k'  \Omega  1_{\ell} + O_p(n\bar \theta \sqrt{\log(n)} + n^{2}r_n \bar \theta^2\big) 
\end{align}
where the high probability is at least $1- o(n^{-3})$. Here to get the RHS, we used the following estimates which can be obtained by employing Bernstein inequality, 
\begin{align*}
& |1_k'  W  1_{\ell}  |\leq C\big( \sqrt{\sum_{i\in \mathcal C_k, j\in \mathcal C_{\ell}} \theta_i\theta_jP_{k\ell} \log(n) } + \log(n) \big) \leq C\big( n\bar \theta \sqrt{P_{k\ell} \log(n) } + \log(n) \big), \notag\\
& |e_{n,i}'A 1_{\ell}- e_{n,i}'\Omega  1_{\ell}|\leq C\big( \sqrt{\sum_{ j\in \mathcal C_{\ell}} \theta_i\theta_j \log(n) } + \log(n) \big) \leq C\big( \sqrt{n\bar \theta^2 \log(n)} + \log(n) \big), \notag\\
&|e_{n,i}'\Omega  1_{\ell}|\leq Cn\bar\theta^2\,. 
\end{align*}
simultaneously for all $1\leq i\leq n$ and $1\leq k, \ell\leq K$, with probability $1- o(n^{-3})$. Note that the second and third inequalities also imply that $\max_i (e_{n,i}' A{\bf1}_n) \leq C \bar\theta^2$ and furthermore,  
\begin{align*}
& \big| (\hat 1_k - 1_k)'  A 1_{\ell} \big|  \leq \|\hat 1_k - 1_k\|_1 \max_i (e_{n,i}' A{\bf1}_n) \leq C n^{2}r_n \bar \theta^2\notag\\
&  |  (\hat 1_k - 1_k)' A  (\hat 1_{\ell} - 1_{\ell})| \leq \|\hat 1_k - 1_k\|_1 \max_i (e_{n,i}' A{\bf1}_n) \leq C n^{2}r_n \bar \theta^2
\end{align*}

Next for denominator, we first have 
\begin{align*}
& \quad \hat 1_k' \widehat \Theta  ({\bf 1}_n{\bf 1}_n' - A) \widehat \Theta \hat 1_{\ell}  \notag\\
&=  1_k' \widehat \Theta  ({\bf 1}_n{\bf 1}_n' - A) \widehat \Theta  1_{\ell}  + (\hat 1_k -1_k)' \widehat \Theta  ({\bf 1}_n{\bf 1}_n' - A) \widehat \Theta  \hat 1_{\ell} + 1_k' \widehat \Theta  ({\bf 1}_n{\bf 1}_n' - A) \widehat \Theta  ( \hat 1_{\ell}  - 1_{\ell}) \notag\\
& = 1_k' \widehat \Theta  ({\bf 1}_n{\bf 1}_n' - A) \widehat \Theta  1_{\ell}  + O_p(n^{2}r_n \bar \theta^2) 
\end{align*}
Write $\Delta = \widehat \Theta -  {\rm diag}(\widehat \Pi P \Pi')\Theta$. We further derive 
\begin{align*}
1_k' \widehat \Theta  ({\bf 1}_n{\bf 1}_n' - A) \widehat \Theta  1_{\ell}  = & 1_k' {\rm diag}(\widehat \Pi P \Pi)\Theta  ({\bf 1}_n{\bf 1}_n' - A)\Theta  {\rm diag}(\widehat \Pi P \Pi) 1_{\ell}   \notag\\
 & +1_k' \Delta ({\bf 1}_n{\bf 1}_n' - A) \widehat \Theta  1_{\ell}   + 1_k' {\rm diag}(\widehat \Pi P \Pi)\Theta  ({\bf 1}_n{\bf 1}_n' - A) \Delta 1_{\ell}   \notag\\
 & = : J_1 + J_2 + J_3
\end{align*}
We analyze each term on the RHS as follows. 
\begin{align*}
J_1 &=  1_k' {\rm diag}(\widehat \Pi P \Pi)\Theta  ({\bf 1}_n{\bf 1}_n' - A)\Theta  {\rm diag}(\widehat \Pi P \Pi) 1_{\ell} \notag\\
&=  1_k' {\rm diag}(\widehat \Pi P \Pi)\Theta  ({\bf 1}_n{\bf 1}_n' - \Omega )\Theta  {\rm diag}(\widehat \Pi P \Pi) 1_{\ell} -  1_k' {\rm diag}(\widehat \Pi P \Pi)\Theta  W\Theta  {\rm diag}(\widehat \Pi P \Pi) 1_{\ell}
\end{align*}
where with probability $1- o(n^{-3})$, since ${\bf 1}_n{\bf 1}_n' - \Omega  ={\bf 1}_n{\bf 1}_n'  - \widetilde \Omega\circ N  = N$, 
\begin{align*}
& \quad 1_k' {\rm diag}(\widehat \Pi P \Pi)\Theta  N\Theta  {\rm diag}(\widehat \Pi P \Pi) 1_{\ell}  \notag\\
& =  1_k' \Theta  N\Theta 1_{\ell}  +  1_k' ({\rm diag}(\widehat \Pi P \Pi)- I_n)\Theta  N\Theta  {\rm diag}(\widehat \Pi P \Pi) 1_{\ell} + 1_k' \Theta  N \Theta({\rm diag}(\widehat \Pi P \Pi) - I_n ) 1_{\ell}  \notag\\
& = \sum_{i\in \mathcal C_k, j\in \mathcal C_{\ell} } \theta_i \theta_j N_{ij} + O_p\big(n^{2}r_n\bar \theta^2 \big)  \asymp n^2 \bar \theta^2
\end{align*}
and 
\begin{align*}
&\big| 1_k' {\rm diag}(\widehat \Pi P \Pi)\Theta  W\Theta  {\rm diag}(\widehat \Pi P \Pi) 1_{\ell} \big| \leq \|W\|  \| 1_k' {\rm diag}(\widehat \Pi P \Pi)\Theta\|^2 \leq n\bar \theta^2 \sqrt{n\bar \theta^2}\, . 
\end{align*}
The last step of the above inequality is due to the non-asymptotic theory of random matrix which gives  $\|W\| \leq \sqrt{n\bar \theta^2}$ with high probability. As a result, 
\[
J_1 = \sum_{i\in \mathcal C_k, j\in \mathcal C_{\ell} } \theta_i \theta_j N_{ij} + O_p\big(n^{2}r_n\bar \theta^2 \big) + O_p\big( (n\bar \theta^2)^{3/2}\big) \, . 
\]

To proceed, we note that 
\begin{align*}
\|1_k' \Delta\|_1 &= \sum_{i\in \mathcal C_k : \hat \pi_i = \pi_i} | \hat \theta_i - \theta_i  | +  \sum_{i\in \mathcal C_k: \hat \pi_i \neq  \pi_i} | \hat \theta_i - \hat \pi_i' P \pi_i  \theta_i |  \notag\\
&\leq Cn \Big(\sqrt{\frac{\log(n) }{n}} + r_n/\bar \theta \Big) + Cnr_n  \Big( \bar \theta \Big(\frac{\log(n) }{n\bar \theta^2}\Big)^{1/4} +\sqrt{r_n}  \Big) \notag\\
& \leq C\sqrt{n\log(n)} + Cnr_n/\bar \theta
\end{align*}

For $J_2$,
\begin{align*}
|J_2| = | 1_k' \Delta ({\bf 1}_n{\bf 1}_n' - A) \widehat \Theta  1_{\ell} |\leq C\|1_k' \Delta\|_1 n\bar \theta \leq Cn \sqrt{n\bar\theta ^2\log(n)} + n^{2}r_n \, .
\end{align*}
For $J_3$, 
\begin{align*}
|J_3| = | 1_k' {\rm diag}(\widehat \Pi P \Pi)\Theta  ({\bf 1}_n{\bf 1}_n' - A) \Delta 1_{\ell}  |\leq C\|1_{\ell}' \Delta\|_1 n\bar \theta \leq Cn \sqrt{n\bar\theta ^2\log(n)} + n^{2}r_n\, .
\end{align*}
Consequently, 
\begin{align*}
\hat 1_k' \widehat \Theta  ({\bf 1}_n{\bf 1}_n' - A) \widehat \Theta \hat 1_{\ell}  =\sum_{i\in \mathcal C_k, j\in \mathcal C_{\ell} } \theta_i \theta_j N_{ij} + O_p\Big(n\sqrt{n\bar\theta ^2\log(n)} + n^{2}r_n 
\Big)
\end{align*}
This, with (\ref{est:P_num}), gives rise to 
\begin{align*}
\widehat{P}_{k\ell}  &= \frac{\sum_{i\in \mathcal C_k, j\in \mathcal C_{\ell} } \theta_i \theta_j N_{ij}P_{k\ell} + O_p(n\bar \theta \sqrt{\log(n)} + n^{2}r_n \bar \theta^2 \big) } {\sum_{i\in \mathcal C_k, j\in \mathcal C_{\ell} } \theta_i \theta_j N_{ij} + O_p\Big(n \sqrt{n\bar\theta ^2\log(n)} + n^{2}r_n \Big)}\notag\\
 & = P_{k\ell} + O_p\Big(  \sqrt{\log(n)} /(n\bar \theta)  + r_n + \sqrt{\log(n)/n\bar \theta^2} + r_n/\bar \theta^2\Big) \notag\\
 & = P_{k\ell} + O_p\Big(    \sqrt{\log(n)/n\bar\theta^2} + r_n/\bar \theta^2\Big) 
\end{align*}
simultaneously for all $1\leq k \neq \ell\leq K$, where the probability is at least $1- o(n^{-3})$.  This completes the proof. 

\subsection{The error rate of $N$} \label{subsec:refit_N}
In this subsection, we prove bounds for $\|N\oslash \widehat{N} - {\bf 1}_n {\bf 1}_n'\|_{F}$ and $\|(N\oslash \widehat{N} - {\bf 1}_n {\bf 1}_n') \circ \widetilde \Omega\|$ under the assumptions in Theorem~\ref{thm:main}. The results are provided as below.

\begin{lemma} \label{lem:error_N}
Suppose the assumptions in Theorem~\ref{thm:main} hold. Then, 
\begin{align*}
&\|(N \oslash \widehat N  - {\bf 1}_n{\bf 1}_n') \|_F  \leq C\Big( \sqrt{n\bar \theta^2 \log(n) } + n r_n + n\bar \theta^2 \sqrt{r_n}\Big) \ll \lambda_{1}(\widetilde \Omega)  \notag\\
&\|(N \oslash \widehat N  - {\bf 1}_n{\bf 1}_n') \circ \widetilde \Omega\|  \leq C\Big(\bar \theta^2 \sqrt{n\bar \theta^2 \log(n) } + n\bar \theta^2 r_n + n\bar \theta^4 \sqrt{r_n}\Big)\ll |\lambda_{K}(\widetilde \Omega)|
\end{align*}
with probability $1- o(n^{-3})$.
\end{lemma}

We prove Lemma~\ref{lem:error_N} below.

By definition, 
\begin{align*}
N\oslash \widehat{N} - {\bf 1}_n {\bf 1}_n' = (\widehat \Theta \widehat \Pi \widehat P \widehat \Pi'\widehat \Theta  - \Theta \Pi P \Pi'\Theta )\circ N
\end{align*}
It follows that
\begin{align*}
& \quad e_{n,i}'\big( N\oslash \widehat{N}  - {\bf 1}_n {\bf 1}_n'\big)e_{n,j} \notag\\
& \leq e_{n,i}'\big(  \widehat \Theta \widehat \Pi \widehat P \widehat \Pi'\widehat \Theta  - \Theta \Pi P \Pi'\Theta \big)e_{n,j} \notag\\
& \leq  e_{n,i}' \big( \widehat \Theta \widehat \Pi (\widehat P- P)  \widehat \Pi'\widehat \Theta \big)e_{n,j} + e_{n,i}'\big( \widehat \Theta (\widehat \Pi - \Pi) P \widehat \Pi'\widehat \Theta \big)e_{n,j} + e_{n,i}'\big( \widehat \Theta \Pi P (\widehat \Pi - \Pi)'\widehat \Theta \big)e_{n,j} \notag\\
& \quad + e_{n,i}'\big( \widehat \Theta \Pi P  \Pi'\widehat \Theta- \Theta \Pi P \Pi'\Theta \big)e_{n,j} 
\end{align*} 
By the error rates for refitting $\theta$ and  $P$, i.e., 
\begin{align*}
\|\widehat P - P \|_{\max}\leq C \Big( \sqrt{\frac{\log(n) }{n\bar \theta^2}} + \frac{r_n}{\bar \theta^2}\Big), \qquad |\hat \theta_ i- \theta_i | \leq C \Big( \sqrt{\frac{\log(n) }{n} }+ \frac{r_n}{\bar \theta}\Big), \quad \text{if $\hat \pi_i  = \pi_i$}
\end{align*}
and 
\[
 |\hat \theta_ i-P_{kk_0} \theta_i | \leq C \Big( \bar\theta (\log(n) /n\bar \theta^2)^{1/4}+ \sqrt{r_n} \Big), \quad \text{if $\hat \pi_i  =e_k, \pi_i = e_{k_0}, k\neq k_0$}
\]
We can derive that with probability $1- o(n^{-3)}$, simultaneously for all $1\leq i, j\leq n$, 
\begin{align}\label{2024093001}
& \big| e_{n,i}' \big( \widehat \Theta \widehat \Pi (\widehat P- P)  \widehat \Pi'\widehat \Theta \big)e_{n,j} \big| = \left\{
\begin{array}{ll}
 0 & \text{if $\hat \pi_i= \hat \pi_j$}\\
 O_p\Big(  \bar \theta\sqrt{\log(n)/n}  +  r_n \Big)  & \text{if $\hat \pi_i\neq  \hat \pi_j$}
\end{array}
\right. \notag\\
&\big| e_{n,i}'\big( \widehat \Theta (\widehat \Pi - \Pi) P \widehat \Pi'\widehat \Theta \big)e_{n,j} \big| =  \left\{
\begin{array}{ll}
 0 & \text{if $\hat \pi_i = \pi_i $}\\
 O_p\big(  \bar \theta^2\big)  & \text{if $\hat \pi_i \neq  \pi_i$}
\end{array}
\right. \notag\\
&\big| e_{n,i}'\big( \widehat \Theta \Pi P (\widehat \Pi - \Pi)'\widehat \Theta \big)e_{n,j}\big| =  \left\{
\begin{array}{ll}
 0 & \text{if $\hat \pi_j = \pi_j $}\\
 O_p\big(  \bar \theta^2\big)  & \text{if $\hat \pi_j \neq  \pi_j$}
\end{array}
\right.
\end{align}
and 
\begin{align}\label{2024093002}
\big|  e_{n,i}'\big( \widehat \Theta \Pi P  \Pi'\widehat \Theta- \Theta \Pi P \Pi'\Theta \big)e_{n,j}  \big| =  \left\{
\begin{array}{ll}
 O_p\Big(\bar\theta \sqrt{\log(n)/n} + r_n \Big) & \text{if $\hat \pi_i = \pi_i, \hat \pi_j = \hat \pi_j $}\\
 O_p\big(  \bar \theta^2\big)  & \text{if $\hat \pi_i \neq  \pi_i$ or $\hat \pi_j \neq  \pi_j$}
\end{array}
\right.
\end{align}
Here we used the fact that $\widehat P _{kk} = P_{kk} = 1$. 

Combining the above estimates together, we obtain that 
\begin{align*}
(N\oslash \widehat{N}  - {\bf 1}_n {\bf 1}_n' )_{ij}=  \left\{
\begin{array}{ll}
 O_p\Big(\bar\theta \sqrt{\log(n)/n} + r_n \Big) & \text{if $\hat \pi_i = \pi_i, \hat \pi_j = \hat \pi_j $}\\
 O_p\big(  \bar \theta^2\big)  & \text{if $\hat \pi_i \neq  \pi_i$ or $\hat \pi_j \neq  \pi_j$}
\end{array}
\right.
\end{align*}
Therefore, in light of  $\|\widehat \Pi - \Pi\|_1 \leq nr_n$, it yields that
\begin{align*}
\| N\oslash \widehat{N}  - {\bf 1}_n {\bf 1}_n' \|_{F} &=
 \sqrt{\sum_{i,j: \hat \pi_i = \pi_i, \hat \pi_j =\pi_j}  ( (N\oslash \widehat{N} )_{ij}   - 1)^2 + \sum_{i,j: \hat \pi_i \neq  \pi_i \text{ or } \hat \pi_j \neq \pi_j}  ((N\oslash \widehat{N} )_{ij}  - 1)^2}
\notag\\
& \leq C\sqrt{n^2\big(\bar\theta \sqrt{\log(n)/n} + r_n\big)^2 + n^{2}r_n \bar \theta^4}  \notag\\
& \leq C\big(\bar\theta \sqrt{n\log(n)} + nr_n + n\bar \theta^2\sqrt{r_n} \,  \big)
\end{align*}
with probability $1- o(n^{-3})$. Due to the conditions that 
\[
r_n\ll \bar \theta^2 \to 0, \qquad n\bar\theta^2 \gg \log(n),
\]
we easily see that 
\[
\| N\oslash \widehat{N}  - {\bf 1}_n {\bf 1}_n' \|_{F}\leq C\big(\bar\theta \sqrt{n\log(n)} + nr_n + n\bar \theta^2\sqrt{r_n} \,  \big) \ll \lambda_1(\widetilde \Omega)
\]
since $\lambda_1(\widetilde \Omega)\asymp n\bar \theta^2$.

Next, we consider $\|(N \oslash \widehat N  - {\bf 1}_n{\bf 1}_n') \circ \widetilde \Omega\|$.
 Note that 
\[
N\oslash \widehat{N}  - {\bf 1}_n {\bf 1}_n' = \big( \widehat \Theta \widehat \Pi \widehat P \widehat \Pi'\widehat \Theta  - \Theta \Pi P \Pi'\Theta  \big) \circ N \,. 
\]
We consider $\| \big( \widehat \Theta \widehat \Pi \widehat P \widehat \Pi'\widehat \Theta  - \Theta \Pi P \Pi'\Theta  \big) \circ \widetilde \Omega \circ N \|$ instead. Since the rank of $\widehat \Theta \widehat \Pi \widehat P \widehat \Pi'\widehat \Theta  - \Theta \Pi P \Pi'\Theta $ is at most $2K$, we bound 
\begin{align*}
\| (N\oslash \widehat{N}  - {\bf 1}_n {\bf 1}_n')  \circ \widetilde \Omega \| & \leq \sqrt{ \sum_{i,j} (N\circ \widetilde \Omega)_{ij}^2 \big( \widehat \Theta \widehat \Pi \widehat P \widehat \Pi'\widehat \Theta  - \Theta \Pi P \Pi'\Theta  \big)_{ij}^2 }  \notag\\
&\leq \|N\circ \widetilde \Omega\|_{\max} \| \widehat \Theta \widehat \Pi \widehat P \widehat \Pi'\widehat \Theta  - \Theta \Pi P \Pi'\Theta  \|_F \notag\\
& \leq \sqrt{2K} \|N\circ \widetilde \Omega\|_{\max}\| \widehat \Theta \widehat \Pi \widehat P \widehat \Pi'\widehat \Theta  - \Theta \Pi P \Pi'\Theta  \|\notag\\
& \leq C \bar \theta^2 \| \widehat \Theta \widehat \Pi \widehat P \widehat \Pi'\widehat \Theta  - \Theta \Pi P \Pi'\Theta  \|
\end{align*}
To proceed, we study the upper bound of $\| \widehat \Theta \widehat \Pi \widehat P \widehat \Pi'\widehat \Theta  - \Theta \Pi P \Pi'\Theta  \|$. Note that 
\[
\| \widehat \Theta \widehat \Pi \widehat P \widehat \Pi'\widehat \Theta  - \Theta \Pi P \Pi'\Theta  \|\leq \| \widehat \Theta \widehat \Pi \widehat P \widehat \Pi'\widehat \Theta  - \Theta \Pi P \Pi'\Theta  \|_{F} \leq \sqrt {2K }  \| \widehat \Theta \widehat \Pi \widehat P \widehat \Pi'\widehat \Theta  - \Theta \Pi P \Pi'\Theta  \|
\]
It suffices to study the upper bound of $\| \widehat \Theta \widehat \Pi \widehat P \widehat \Pi'\widehat \Theta  - \Theta \Pi P \Pi'\Theta  \|_{F}$, which by previous arguments (\ref{2024093001}) and (\ref{2024093002}), is given by 
\[
\| \widehat \Theta \widehat \Pi \widehat P \widehat \Pi'\widehat \Theta  - \Theta \Pi P \Pi'\Theta  \|_{F} \leq C\big(\bar\theta \sqrt{n\log(n)} + nr_n + n\bar \theta^2\sqrt{r_n} \,  \big)
\]

We thus conclude that 
\[
\| (N\oslash \widehat{N}  - {\bf 1}_n {\bf 1}_n')  \circ \widetilde \Omega \| \leq C\big(\bar\theta^2 \sqrt{n\bar\theta^2\log(n)} + n \bar \theta^2 r_n + n\bar \theta^4\sqrt{r_n} \,  \big)
\]
Our assumptions in Theorem~\ref{thm:main} says   that 
\[
\sqrt{n\bar \theta^2} \lambda_{\min}(P )\geq C\log(n), \quad  r_n\ll \lambda_{\min}(P ), \quad r_n \ll \lambda_{\min}^2(P) /\bar \theta^4
\]
Using these conditions, together with $|\lambda_{K}(\widetilde \Omega) | \asymp n\bar \theta^2 \lambda_{\min} (P)$, we easily derive that 
\[
\| (N\oslash \widehat{N}  - {\bf 1}_n {\bf 1}_n')  \circ \widetilde \Omega \| \leq C\big(\bar\theta^2 \sqrt{n\bar\theta^2\log(n)} + n \bar \theta^2 r_n + n\bar \theta^4\sqrt{r_n} \,  \big)\ll |\lambda_{K}(\widetilde \Omega)|
\]
This finishes the proof of error rates for $\widehat N$.

\subsection{Proof of Theorem~\ref{thm:main}}\label{subsec:main_thm} 
We now prove Theorem~\ref{thm:main} using the results in Sections~\ref{subsec:refit_P} and \ref{subsec:refit_N}, and Lemma~\ref{lemma:RSCORE}. 

To begin with, we verify the additional conditions in Lemma~\ref{lemma:RSCORE} (i.e., (\ref{asm:rscore1a} ) and (\ref{asm:rscore1b} ))  under the assumptions in Theorem~\ref{thm:main}. Notice that $r_n \asymp \delta_n$ where $\delta_n = \max\{\|(N - {\bf 1}_n {\bf 1}_n') \circ 
\widetilde{\Omega}\|^2, \lambda_1(\widetilde{\Omega})\} / \lambda_K^2(\widetilde{\Omega})$.
Specifically, in Section~\ref{subsec:refit_P}, we have shown that 
\[
\|\widehat{P} - P\|_{\max} \leq C \Big( \sqrt{\frac{\log(n) }{n\bar \theta^2}} + \frac{r_n}{\bar \theta^2}\Big)
\]
From the assumptions in   Theorem~\ref{thm:main}, we have 
\[
\sqrt{n\bar \theta^2} \lambda_{\min}(P) \geq C\log(n), \qquad r_n \ll \min\{|\lambda_{\min}(P)|\bar \theta, \bar \theta^2\}
\]
It follows that $\|\widehat{P} - P\|_{\max} = o(1)$ and 
\[
 \sqrt{\frac{\log(n) }{n\bar \theta^2}}  |\lambda_{\min} (P) |^{-1}\bar \theta = o(1), \qquad  \frac{r_n}{\bar \theta^2} |\lambda_{\min} (P) |^{-1}\bar \theta = o(1)\,. 
\]
Next, thanks to $r_n\asymp \delta_n \ll\lambda_{\min}^2(P) /\bar \theta^2$,
\[
\|\widehat \Pi - \Pi\| (\sqrt{n}\, |\lambda_{\min}(P)|)^{-1} \bar\theta \leq C\sqrt{nr_n}(\sqrt{n}\, |\lambda_{\min}(P)|)^{-1} \bar\theta \ll 1\,. 
\]
Therefore, the conditions in (\ref{asm:rscore1a}) are satisfied.  It was mentioning that in Section~\ref{subsec:refit_N}, we have validated (\ref{asm:rscore1b}). Lastly, by Lemma~\ref{lem:error_theta}, and the conditions that $n\bar \theta^2\gg \log(n), r_n\ll \bar \theta^2$, we easily see that $\hat \theta_i < C \bar \theta$.

Therefore, we can apply the results in Lemma~\ref{lemma:RSCORE}, which gives that 

\begin{align*}
r_n(\widehat{\Pi}^{rscore})  \leq   \frac{C[\| (N\oslash \widehat N - {\bf 1}_n{\bf 1}_n')  \circ \widetilde \Omega\|^2  +\tau_n^2+   { \lambda_1(\widetilde \Omega)}]}{\big| \lambda_K(\widetilde \Omega) \big|^2}
\end{align*}
where 
$\tau_n =   \sqrt{n}\bar \theta^3 [\sqrt{n} \| \widehat P - P\|_{\max}  +   \| \widehat{\Pi}^{score}- \Pi\|]$.

Furthermore, we plug in the upper bounds of  $ \| \widehat P - P\|_{\max}$ and $\| (N\oslash \widehat N - {\bf 1}_n{\bf 1}_n')  \circ \widetilde \Omega\| $ in  Sections~\ref{subsec:refit_P} and \ref{subsec:refit_N} and note that $\|\widehat{\Pi}^{score} - \Pi\| \leq C \sqrt{nr_n }\leq C \sqrt{n\delta_n}$. Elementary computations lead to 
\[
[\| (N\oslash \widehat N - {\bf 1}_n{\bf 1}_n')  \circ \widetilde \Omega\|^2  +\tau_n^2\leq C\Big( n\bar \theta^4 \log(n) + n^2\bar \theta^2 \delta_n^2 + n^2 \bar \theta^6 \delta_n\Big) 
\]
Thereby, we conclude the proof of Theorem~\ref{thm:main}.

\subsection{Proof of Corollary~\ref{cor:main}}
Using the assumptions that $\lambda_{\min} (P) \geq C$ for a constant $C > 0$, we see that the condition of $\delta_n$ in  Theorem~\ref{thm:main} is reduced to 
\[
\delta_n\ll \bar \theta^2
\]
Also, we can derive 
\[
\|(N - {\bf 1}_n {\bf 1}_n') \circ \widetilde{\Omega}\|^2 / \lambda_K^2(\widetilde{\Omega})\leq \|(N - {\bf 1}_n {\bf 1}_n')\|_{\max}^2 \|\widetilde \Omega\|^2 / \lambda_K^2(\widetilde{\Omega})\leq C\bar \theta^4 \ll \bar \theta^2\to 0 
\]
and 
\[
\lambda_1(\widetilde{\Omega})/ \lambda_K^2(\widetilde{\Omega}) \leq \frac{1}{n\bar\theta^2}\ll \bar \theta^2\,. 
\]
by the assumption that $n \bar{\theta}^4  \to \infty$.
Therefore, by Theorem~\ref{thm:SCORE}, we obtain that 
\[
\delta_n \asymp r_n(\widehat{\Pi}^{score})  \leq C \Big( \frac{1}{n\bar \theta^2} + \bar \theta^4\Big) 
\]
and $\delta_n \ll \bar \theta^2$. Then, the conditions in Theorem~\ref{thm:main} are satisfied. Therefore, 
\begin{align*}
r_n(\widehat{\Pi}^{rscore}) & \leq \frac{C}{\lambda_K^2(\widetilde \Omega)} \Big( \lambda_1(\widetilde{\Omega}) + n\bar \theta^4 \log(n) + n^2 \bar \theta^2 \delta_n^2 + n^2 \bar \theta^6 \delta_n\Big)\notag\\
&\leq \frac{C}{n^2\bar \theta^4} \Big( n\bar \theta^2 +n\bar \theta^4 \log(n) + n^2 \bar \theta^2[1/( n\bar \theta^2)^2 + \bar \theta^8] + n^2 \bar \theta^6 [1/ n\bar \theta^2 + \bar \theta^4] \Big) \notag\\
& \leq C\Big( \frac{1}{n\bar \theta^2} + \bar\theta^6  +   \frac{\log(n)}{n}\Big)
\end{align*}
where we used $\bar \theta = o(1)$ and $n\bar \theta^4\to \infty$. We thus complete the proof of
Corollary~\ref{cor:main}.

\bibliography{network}
\bibliographystyle{iclr2025_conference}

\end{document}